\documentclass[oneside,reqno]{amsart}
\usepackage{amssymb}
\usepackage{eucal}
\usepackage{color}
\usepackage[pdfpagelabels, colorlinks, citecolor=black, linkcolor=black, urlcolor=black, pdftitle=The\ Lie\ Algebraic\ Significance\ of\ Symmetric\ Informationally\ Complete\ Measurements]{hyperref}
\usepackage[square,comma,numbers,sort&compress]{natbib}
\usepackage{hypernat}
\newtheorem{theorem}{Theorem}
\newtheorem{corollary}[theorem]{Corollary}
\newtheorem{lemma}[theorem]{Lemma}
\theoremstyle{remark}
\newtheorem*{remark}{Remark}

\newcommand{\be}{\begin{equation}}
\newcommand{\ee}{\end{equation}}
\newcommand{\bc}{\begin{center}}
\newcommand{\ec}{\end{center}}
\newcommand{\bmt}{\begin{pmatrix}}
\newcommand{\emt}{\end{pmatrix}}

\newcommand{\dr}{\rangle\!\rangle}
\newcommand{\dl}{\langle\!\langle}
\newcommand{\ggrm}{P}
\newcommand{\grmJ}{J_P}

\newcommand{\pty}{U_{\mathrm{P}}}
\newcommand{\ptyD}{U^{\vphantom{\dagger}}_{\mathrm{P}}}
\newcommand{\ev}{e}
\newcommand{\vv}{v_0}
\newcommand{\gv}{g}
\newcommand{\hv}{h}
\newcommand{\av}{a}
\newcommand{\bv}{b}

\newcommand{\ebv}{\bar{e}}
\newcommand{\vpu}[1]{^{\vphantom{#1}}}
\newcommand{\ph}[1]{\phantom{#1}}
\newcommand{\mt}{K}
\newcommand{\str}{C}
\newcommand{\stra}{\tilde{C}}
\newcommand{\mtb}{M}
\newcommand{\lie}{L}
\newcommand{\liet}{l}
\newcommand{\cls}{\kappa}
\newcommand{\liq}{P}
\newcommand{\slGen}{B}
\newcommand{\ad}{\mathrm{ad}}
\newcommand{\sign}{\epsilon}
\DeclareMathOperator{\rel}{Re}

\DeclareMathOperator{\Tr}{Tr}

\DeclareMathOperator{\rnk}{rank}

\DeclareMathOperator{\ua}{u}
\DeclareMathOperator{\gl}{gl}
\DeclareMathOperator{\sla}{sl}

\begin{document}

\begin{titlepage}
\begin{center}
\bfseries {\Large The Lie Algebraic Significance of\\  Symmetric Informationally Complete Measurements}
\end{center}
\vspace{1 cm}
\begin{center} D.M.~Appleby, Steven T.~Flammia and Christopher A.~Fuchs
\end{center}
\begin{center} Perimeter Institute for Theoretical Physics \\ Waterloo, Ontario N2L 2Y5, Canada
 \end{center}

\vspace{.2 in}

\begin{center} December 30, 2009
\end{center}
\vspace{1.5 cm}

\begin{center} \textbf{Abstract}

\vspace{0.1 in}

\parbox{12 cm }{
Examples of symmetric informationally complete positive operator valued measures (SIC-POVMs) have been constructed in every  dimension $\le 67$.  However, it remains an open question whether they exist in all finite dimensions. A SIC-POVM is usually thought of as a highly symmetric structure in quantum state space.  However, its elements can equally well be regarded as a basis for the Lie algebra $\gl(d,\mathbb{C})$.  In this paper we examine the resulting structure constants, which are calculated from  the traces of the triple products of the SIC-POVM elements and which, it turns out, characterize the SIC-POVM up to unitary equivalence.  We show that the structure constants have numerous remarkable properties.   In particular we show that the existence of a SIC-POVM in dimension $d$ is equivalent to the existence of a certain structure in the adjoint representation of $\gl(d,\mathbb{C})$.  We hope that transforming the problem in this way, from a question about quantum state space to a question about Lie algebras, may help to make the existence problem tractable. 
}
\end{center}
\vspace{.9 cm}

\begin{center}
\parbox{12 cm }{
\setcounter{tocdepth}{1}
\tableofcontents
}
\end{center}

\end{titlepage}

{\allowdisplaybreaks
\section{Introduction}
Symmetric informationally complete positive operator-valued measures (SIC-POVMs) present us with what is, simultaneously,  one of the most interesting, and one of the most difficult and tantalizing problems in quantum information~\cite{Hoggar,Zauner,Caves,FuchsA,RBSC,SanigaPlanatRosu,FuchsB,Rehacek,WoottersA,Ingemar1,GrasslA,Renes,Ingemar2,Koenig,ZimanBuzek,selfA,KlappA,KlappB,GrasslB,Ballester,Gross,Colin,GodsilRoy,Howard,Scott,DurtA,Flammia,GrasslC,Kim,selfB,WaldronBos,RoyScott,KiblerA,selfC,Khat,Bodmann,KiblerB,Ingemar3,GrasslD,GrasslE,Fickus,selfD,FuchsSchack,selfE,Ericcson,ScottGrassl}. SIC-POVMs are important practically, with applications to quantum tomography and cryptography~\cite{FuchsA,Rehacek,Renes,ZimanBuzek,Ballester,Kim}, and to classical signal processing~\cite{Howard,Bodmann}.  However, without in any way wishing to impugn the significance of  the applications which have so far been proposed,  it appears to us that the interest of SIC-POVMs stems less from these particular proposed uses than from rather broader, more general considerations:  the sense one gets that SICs are telling us something deep, and hitherto unsuspected about the structure of quantum state space.  In spite of its being the central object about which the rest of quantum mechanics rotates, and notwithstanding the efforts of numerous investigators~\cite{BengtssonBook}, the geometry of quantum state space continues to be surprisingly ill-understood.  The hope which inspires our  efforts  is that a solution to the SIC problem will prove to be the key, not just to SIC-POVMs narrowly conceived, but to the geometry of state space in general. Such things are, by nature, unpredictable.  However,  it is not unreasonable to speculate that a better theoretical understanding of the geometry of quantum state space might have important practical consequences:  not only the applications listed above, but perhaps other applications which have yet to be conceived.   On a more foundational level one may hope that it will lead to a much improved understanding of the conceptual message of quantum mechanics~\cite{FuchsB,FuchsSchack,Ericcson,FuchsC}.   

Having said why we describe the problem as interesting, let us now explain why we describe it as tantalizing.  The trouble is that, although there is an abundance of reasons for suspecting that SIC-POVMs exist in every finite dimension (exact and high-precision numerical examples~\cite{Hoggar,Zauner,RBSC,GrasslA,selfA,GrasslB,GrasslC,GrasslD,ScottGrassl} having now been constructed in every dimension up to $67$), and in spite of the intense efforts of many people~\cite{Hoggar,Zauner,Caves,FuchsA,RBSC,SanigaPlanatRosu,FuchsB,Rehacek,WoottersA,Ingemar1,GrasslA,Renes,Ingemar2,Koenig,ZimanBuzek,selfA,KlappA,KlappB,GrasslB,Ballester,Gross,Colin,GodsilRoy,Howard,Scott,DurtA,Flammia,GrasslC,Kim,selfB,WaldronBos,RoyScott,KiblerA,selfC,Khat,Bodmann,KiblerB,Ingemar3,GrasslD,GrasslE,Fickus,selfD,FuchsSchack,selfE,Ericcson,ScottGrassl} extending over a period of more than ten years, a general existence proof continues to elude us.     In their seminal paper on the subject, published in $2004$, Renes \emph{et al}~\cite{RBSC} say ``A rigorous proof of existence of SIC-POVMs in all finite dimensions seems tantalizingly close, yet remains somehow distant.''   They could have said the same if they were writing today.

The purpose of this paper is to try to take our understanding of SIC mathematics (as it might be called) a little further forward.  
The research we report began with a chance numerical discovery made while we were working on a different problem.  Pursuing that initial numerical hint we uncovered a rich and interesting  set of connections between SIC-POVMs in dimension $d$ and the Lie Algebra  $\gl(d,\mathbb{C})$. The existence of these connections came as a surprise to us.  However, in retrospect it is, perhaps, not so surprising.  Interest in SIC-POVMs has, to date,  focused on the fact that an arbitrary density matrix can be expanded in terms of  a SIC-POVM.  However, a SIC-POVM in dimension $d$ does in fact provide a basis, not  just for the space of density matrices, but for the space of all $d\times d$ complex matrices---\emph{i.e.} the Lie algebra $\gl(d,\mathbb{C})$.  Boykin \emph{et al}~\cite{boykin} have recently shown that there is a connection between the existence problem for maximal sets of MUBs (mutually unbiased bases) and  the theory of Lie algebras.  Since SIC-POVMs share with MUBs the property of being highly symmetrical structures in quantum state space it might have been anticipated that there are also some interesting connections between SIC-POVMs and Lie algebras.

Our main result (proved in Sections~\ref{sec:QQTProp1}, \ref{sec:QQTpropGen} and~\ref{sec:QQTprop2}) is that the proposition, that a SIC-POVM exists in dimension $d$, is equivalent to a proposition about the adjoint representation of $\gl(d,\mathbb{C})$.  Our hope is that transforming the problem in this way, from a question about quantum state space to a question about Lie algebras, may help to make the SIC-existence problem tractable.  But even if this hope fails to materialize we feel that this result, along with the many other results we obtain, provides some additional insight into these structures.  

In $d$ dimensional Hilbert space $\mathcal{H}_d$ a SIC-POVM is a set of $d^2$ operators $E_1$, \dots ,$E_{d^2}$ of the form
\be
E_r = \frac{1}{d} \Pi_r
\ee
where the $\Pi_r$ are rank-$1$ projectors  with the property
\be
\Tr(\Pi_r\Pi_s) = \begin{cases} 1 \qquad & r=s \\ \frac{1}{d+1} \qquad & r\neq s \end{cases}
\ee
We will refer to the $\Pi_r$  as SIC projectors, and we will say that $\{\Pi_r\colon r=1,\dots, d^2\}$ is a SIC set. 

It follows from this definition that the $E_r$ satisfy 
\be
\sum_{r=1}^{d^2} E_r = I
\ee
(so they constitute a POVM), and that they are linearly independent (so the POVM is informationally complete).

It is an open question whether SIC-POVMs exist for all values of $d$.  However, examples have been constructed analytically in dimensions $2$--$15$ inclusive~\cite{Hoggar,Zauner,GrasslA,selfA,GrasslB,GrasslC,GrasslD,ScottGrassl}, and in dimensions $19$, $24$, $35$ and $48$~\cite{selfA,ScottGrassl}.  Moreover, high precision numerical solutions have been constructed in dimensions $2$--$67$ inclusive~\cite{RBSC,ScottGrassl}.   This lends some plausibility to the speculation that they exist in all dimensions.  For a comprehensive account of the current state of knowledge in this regard, and many new results, see the recent study by Scott and Grassl~\cite{ScottGrassl}.

All known SIC-POVMs have a group covariance property.  In other words, there exists
\begin{enumerate}
\item a  group $G$ having $d^2$ elements
\item a projective unitary representation of $G$ on $\mathcal{H}_d$: \emph{i.e.}\ a map $g\to U_g$ from $G$ to the set of unitaries such that $U_{g_1}U_{g_2} \sim U_{g_1 g_2}$ for all $g_1$, $g_2$ (where the notation ``$\sim$'' means ``equals up to a phase'')
\item a normalized vector $|\psi\rangle$ (the fiducial vector)
\end{enumerate}
such that the SIC-projectors are given by
\be
\Pi^{\vphantom{\dagger}}_g = U^{\vphantom{\dagger}}_{g} |\psi \rangle \langle \psi | U^{\dagger}_g
\ee
(where we label the  projector by the group element $g$, rather than the integer  $r$ as above). 

 Most  known SIC-POVMs  are covariant under the action of the Weyl-Heisenberg group (though not all---see Renes \emph{et al}~\cite{RBSC} and, for an explicit example of a non Weyl-Heisenberg SIC-POVM, Grassl~\cite{GrasslB}).  Here the group is $\mathbb{Z}_d \times \mathbb{Z}_d$, and the projective representation is $\mathbf{p}\to D_{\mathbf{p}}$, where $\mathbf{p} = (p_1,p_2)\in \mathbb{Z}_d \times \mathbb{Z}_d$ and $D_{\mathbf{p}}$ is the corresponding  Weyl-Heisenberg displacement operator
\be
D_{\mathbf{p}} = \sum_{r}^{d-1} \tau^{(2 r+ p_1) p_2} |r+p_1\rangle \langle r|
\ee
In this expression $\tau =  e^{\frac{i\pi (d+1) }{d}}$, the vectors  $|0\rangle, \dots |d-1\rangle$ are an orthonormal basis, and the addition in $|r+p_1\rangle$ is modulo $d$.  For more details see, for example, ref.~\cite{selfA}.

One should not attach too much weight to the fact that all known SIC-POVMs have a group covariance property as this may only  reflect the fact that group covariant SIC-POVMs are much easier to construct. So in this paper we will try to prove as much as we can without assuming such a property.  One potential benefit of this attitude is that, by accumulating enough facts about SIC-POVMs in general, we may eventually get to the point where we can  answer the question, whether  all  SIC-POVMs actually do   have a group covariance property.

The fact that the $d^2$ operators $\Pi_r$ are linearly independent means that they form a basis for the complex Lie algebra $\gl(d,\mathbb{C})$ (the set of all operators acting on $\mathcal{H}_d$).  Since the $\Pi_r$ are Hermitian, then  $i \Pi_r$ forms a basis also for the real Lie algebra $\ua(d)$ (the set of all anti-Hermitian operators acting on $\mathcal{H}_d$).   So for any operator $A \in \gl(d,\mathbb{C})$  there is a unique set of expansion coefficients $a_r$  such that
\be
A = \sum_{r=1}^{d^2} a_r \Pi_r
\ee
To find the expansion coefficients we can use the fact that
\be
\sum_{s=1}^{d^2} \Tr(\Pi_r \Pi_s) \left(\frac{d+1}{d} \delta_{st}-\frac{1}{d^2} \right) = \delta_{rt}
\ee 
from which it  follows 
\be
a_r =  \frac{d+1}{d} \Tr(\Pi_r A) - \frac{1}{d}\Tr(A) \
\label{eq:Aexpansion}
\ee
Specializing to the case $A=\Pi_r\Pi_s$ we find
\be
\Pi_r \Pi_s = \frac{d+1}{d} \left(\sum_{t=1}^{d^2} T_{rst} \Pi_t\right) - \frac{d\delta_{rs}+1}{d+1} I
\label{eq:PirPisProduct}
\ee
where
\be
T_{rst} = \Tr\left( \Pi_r \Pi_s \Pi_t\right)
\label{eq:TripProdDefA} 
\ee
To a large extent this paper consists in an exploration of the properties of these important quantities, which we will  refer to as the triple products.  They are intimately related to the geometric phase, in which context they are usually referred to as $3$-vertex Bargmann invariants (see Mukunda~\emph{et al}~\cite{mukunda}, and references cited therein).      We have, as an immediate consequence of the definition, 
\be
T^{\vphantom{*}}_{rst}=T^{\vphantom{*}}_{trs}=T^{\vphantom{*}}_{str} = T^{*}_{rts} = T^{*}_{tsr} = T^{*}_{srt}
\ee
It is convenient to define
\begin{align}
J^{\vphantom{*}}_{rst} & = \frac{d+1}{ d} (T^{\vphantom{*}}_{rst}-T^{*}_{rst})
\label{eq:JmatrixDef}
\\
R^{\vphantom{*}}_{rst} & = \frac{d+1}{ d} (T^{\vphantom{*}}_{rst}+T^{*}_{rst})
\label{eq:RmatrixDef}
\end{align}
So $J_{rst}$ is imaginary and completely anti-symmetric; $R_{rst}$ is real and completely symmetric.  Both  these quantities play a significant role in the theory.  It follows from Eq.~(\ref{eq:PirPisProduct}) that 
\be
[\Pi_r,\Pi_s] = \sum_{t=1}^{d^2} J_{rst} \Pi_t
\ee
So the $J_{rst}$ are  structure constants for the Lie algebra $\gl(d,\mathbb{C})$.  As an immediate consequence of this they satisfy the Jacobi identity:
\be
\sum_{b=1}^{d^2}\bigl(J_{rsb}J_{tba} + J_{stb}J_{rba} +J_{trb}J_{sba}\bigr) = 0
\ee
for all $r,s,t,a$.  The Jacobi identity holds for any representation of the structure constants.  In the following sections we will derive many other identities which are specific to this particular representation.

Turning to the quantities $R_{rst}$, it follows from Eq.~(\ref{eq:PirPisProduct}) that they feature in the expression for the anti-commutator
\be
\{\Pi_r,\Pi_s\}  =  \sum_{t} R_{rst} \Pi_t - \frac{2(d\delta_{rs}+1)}{d+1}I
\ee
They also play an important role in the description of quantum state space.  Let $\rho$ be any density matrix and let $p_r=\frac{1}{d}\Tr(\Pi_r \rho)$ be the probability of obtaining  outcome $r$ in the measurement described by the POVM with elements $\frac{1}{d} \Pi_r$.  Then it follows from Eq.~(\ref{eq:Aexpansion}) that $\rho$ can be reconstructed from the probabilities by
\be
\rho = \sum_{r=1}^{d^2} \left( (d+1)p_r - \frac{1}{d}\right) \Pi_r
\ee
Suppose, now, that the $p_r$ are \emph{any} set of $d^2$ real numbers.   So we do not assume that the $p_r$ are even probabilities, let alone the probabilities coming from a density matrix according to the prescription $p_r = \frac{1}{d}\Tr(\Pi_r\rho )$.  Then it is shown in ref.~\cite{selfC} that the $p_r$ are  in fact the probabilities coming from a pure state if and only if they satisfy the two conditions
\begin{align}
\sum_{r=1}^{d^2} p^2_r & = \frac{2}{d(d+1)}
\\
\sum_{r,s,t=1}^{d^2} R_{rst} p_rp_sp_t  & = \frac{2(d+7)}{d(d+1)^2}
\end{align}

Let us look at the quantities $J_{rst}$ and $R_{rst}$ in a little more detail.  For each $r$ choose a unit vector $|\psi_r\rangle$ such that $\Pi_r = |\psi_r\rangle \langle \psi_r |$.  Then the Gram matrix for these vectors is of the form
\be
G_{rs} = \langle \psi_r | \psi_s \rangle = \mt_{rs} e^{i\theta_{rs}}
\ee
where the matrix $\theta_{rs}$ is anti-symmetric and 
\be
\mt_{rs} = \sqrt{\frac{d\delta_{rs} +1}{d+1}}
\label{eq:KrsDef}
\ee
  Note that the SIC-POVM does not determine the angles $\theta_{rs}$ uniquely since making the replacements $|\psi_r\rangle \to e^{i\phi_r} |\psi_r \rangle$ leaves the SIC-POVM unaltered, but changes the angles $\theta_{rs}$ according to the prescription $\theta_{rs} \to \theta_{rs}-\phi_r + \phi_s$.  This freedom to rephase the vectors $|\psi_r\rangle$ is not usually important.  However, it sometimes has interesting consequences (see Section~\ref{sec:GGstarProp}).    It can be thought of as a kind of gauge freedom.

The Gram matrix satisfies an important identity.  Every SIC-POVM has the  $2$-design property~\cite{RBSC,KlappA}
\be
\sum_{r=1}^{d^2} \Pi_r \otimes \Pi_r = \frac{2d}{d+1} P_{\mathrm{sym}}
\ee
where $P_{\mathrm{sym}}$ is the projector onto the symmetric subspace of $\mathcal{H}_d\otimes \mathcal{H}_d$.  Expressed in terms of the Gram matrix this becomes
\be
\sum_{r=1}^{d^2} G_{s_1 r}G_{s_2 r} G_{r t_1} G_{r t_2} = 
\frac{d}{d+1} \bigl( G_{s_1t_1}G_{s_2t_2} + G_{s_1t_2}G_{s_2t_1}
\bigr)
\label{eq:2designProp}
\ee

Turning to the triple products we have
\be
T_{rst} = G_{rs} G_{st} G_{tr} 
=
\mt_{rs}\mt_{st}\mt_{tr} e^{i\theta_{rst}}
\label{eq:TripProdDefB} 
\ee
where
\be
\theta_{rst} = \theta_{rs}+\theta_{st}+\theta_{tr} 
\ee
Note that the tensor $\theta_{rst}$ is completely anti-symmetric. In particular $\theta_{rst}=0$ if any two of the indices are the same.   Note also that re-phasing the vectors $|\psi_r\rangle$ leaves the tensors $T_{rst}$ and $\theta_{rst}$ unchanged.  They are in that sense gauge invariant.

Finally, we have the following expressions for  $J_{rst}$ and $R_{rst}$:
\begin{align}
J_{rst} & = \frac{2 i}{d\sqrt{d+1}} \sin \theta_{rst}
\\
R_{rst} & = 
\frac{2(d+1)}{d} \mt_{rs}\mt_{st}\mt_{tr}\cos\theta_{rst}
\end{align}
Like the triple products, $J_{rst}$ and $R_{rst}$ are gauge invariant.

For later reference let us note that the matrix $J_r$, with matrix elements
\be
(J_r)_{st} = J_{rst}
\label{eq:JrMatsDef}
\ee
is the adjoint representative of $\Pi_r$ in the SIC-projector basis:
\be
\ad_{\Pi_r} \Pi_s = [\Pi_r,\Pi_s] = \sum_{t=1}^{d^2} J_{rst} \Pi_t
\ee

It can be seen that all the interesting features of the tensor $G_{rs}$ (respectively, the tensors $T_{rst}$, $J_{rst}$ and $R_{rst}$) are contained in the order-$2$ angle tensor $\theta_{rs}$ (respectively, the order-$3$ angle tensor $\theta_{rst}$).  It is also easy to see that, for any unitary $U$, the transformation 
\be
\Pi_r \to U \Pi_r U^{\dagger}
\ee
leaves the angle tensors invariant.  This suggests that we shift our focus from individual SIC-POVMs to  families of unitarily equivalent SIC-POVMs---SIC-families, as we will call them for short.

We begin our investigation in Section~\ref{sec:AngleTensors} by giving  necessary and sufficient conditions for an arbitrary tensor $\theta_{rs}$ (respectively $\theta_{rst}$) to be the rank-$2$ (respectively rank-$3$) angle tensor corresponding to a SIC-family.  We also show that either angle tensor uniquely determines the corresponding SIC-family.   Finally we describe a method for reconstructing the SIC-family, starting from a knowledge of either of the two angle tensors.

In Sections~\ref{sec:QQTProp1}, \ref{sec:QQTpropGen} and~\ref{sec:QQTprop2} we prove the central result of this paper:  namely, that the existence of a SIC-POVM in dimension $d$ is equivalent to the existence of a certain very special set of  matrices in the adjoint representation of $\gl(d,\mathbb{C})$. In Section~\ref{sec:QQTProp1} we show that, for any SIC-POVM, the adjoint matrices $J_r$ have the spectral decomposition
\be
J^{\vphantom{\mathrm{T}}}_r = Q^{\vphantom{\mathrm{T}}}_{r} - Q^{\mathrm{T}}_{r}
\ee 
where $Q_r$ is a rank $d-1$ projector which has the remarkable property of being orthogonal to its own transpose:
\be
Q^{\vphantom{\mathrm{T}}}_r Q^{\mathrm{T}}_r =  0
\ee
We refer to this feature of the adjoint  matrices as the $Q$-$Q^{\mathrm{T}}$ property. In Section~\ref{sec:QQTProp1} we also show that from a knowledge of the $J$ matrices it is possible to reconstruct the corresponding SIC-family.  
In Section~\ref{sec:QQTpropGen} we characterize the general class of projectors which have the property of being orthogonal to their own transpose. Then, in Section~\ref{sec:QQTprop2}, we prove a converse of the result established in Section~\ref{sec:QQTProp1}.   The $Q$-$Q^{\mathrm{T}}$ property is not completely equivalent to the property of being a SIC set. However, it turns out that it is, in a certain sense, very nearly equivalent.   To be more specific:  let $\lie_r$ be any set of $d^2$ Hermitian operators which constitute a basis for $\gl(d,\mathbb{C})$ and let $\str_r$ be the adjoint representative of $L_r$ in this basis.
Then the necessary and sufficient condition for  the $\str_r$ to have the spectral decomposition
\be
\str^{\vphantom{\mathrm{T}}}_r = Q^{\vphantom{\mathrm{T}}}_{r} - Q^{\mathrm{T}}_{r}
\ee
where $Q_r$ is a rank $d-1$ projector such that $Q^{\vphantom{\mathrm{T}}}_r Q^{\mathrm{T}}_r =  0$ is that there exists a SIC set $\Pi_r$ such that $\lie_r = \sign_r (\Pi_r + \alpha I)$ for some fixed number $\alpha\in\mathbb{R}$ and signs $\sign_r=\pm 1$. In particular, the existence of an Hermitian basis for $\gl(d,\mathbb{C})$ having the $Q$-$Q^{\mathrm{T}}$ property is both necesary and sufficient  for the existence of a SIC-POVM in dimension $d$. 

In Section~\ref{eq:SICSandsldC} we digress briefly, and consider  $\sla(d,\mathbb{C})$ (the Lie algebra consisting of all trace-zero $d\times d$ complex matrices).  As we have explained, this paper is motivated by the hope that a Lie algebraic perspective will cast light on the SIC-existence problem, rather than by an interest in Lie algebras as such.     We focus on $\gl(d,\mathbb{C})$ because that is the case where the connection with SIC-POVMs seems most straightforward.  However a SIC-POVM also gives rise to an interesting geometrical structure in $\sla(d,\mathbb{C})$, as we show in Section~\ref{eq:SICSandsldC}.

In Section~\ref{sec:AdjRepFurther} we derive a number of additional identities satisfied by  the $J$ and $Q$ matrices.

The complex projectors $Q\vpu{T}_r$, $Q^{\mathrm{T}}_r$ and the real projector $Q\vpu{T}_r + Q^{\mathrm{T}}_r$ define three families of subspaces.  It turns out that there are some interesting geometrical relationships between these subspaces, which we study in Section~\ref{sec:geometry}.
  
Finally, in Section~\ref{sec:GGstarProp} we show that, with the appropriate choice of gauge,   the Gram matrix corresponding to a Weyl-Heisenberg covariant SIC-family has a feature analogous to the $Q$-$Q^{\mathrm{T}}$ property, which we call the $\ggrm$-$\ggrm^{\mathrm{T}}$ property.  It is an open question whether this result generalizes to other SIC-families, not covariant with respect to the Weyl-Heisenberg group.

\section{The Angle Tensors}
\label{sec:AngleTensors}
The purpose of this section is to establish the necessary and sufficient conditions  for an arbitrary tensor $\theta_{rs}$ (respectively $\theta_{rst}$) to be the order-$2$ (respectively order-$3$) angle tensor for a SIC-family.  We will also show that either one of the angle tensors is enough to uniquely determine the SIC-family.  Moreover, we will describe  explicit procedures for reconstructing the family,  starting from a knowledge of one of the angle tensors.

We begin by considering the general class of  POVMs  (not just SIC-POVMs) which consist of $d^2$ rank-$1$ elements. A POVM of this type is thus defined by a set of $d^2$ vectors $|\xi_1\rangle , \dots, |\xi_{d^2}\rangle$ with the property
\be
\sum_{r=1}^{d^2}|\xi_r\rangle \langle \xi_r | = I
\ee
Note that $\sum_{r=1}^{d^2} \bigl\| |\xi_r \rangle \bigr\|^2 = d$, so the vectors $|\xi_r\rangle$ cannot all be normalized. In the particular case of a SIC-POVM the vectors all have the same norm $\bigl\| |\xi_r \rangle \bigr\| = \frac{1}{\sqrt{d}}$.  However in the general case they may have different norms.

   Given a set of such vectors consider the Gram matrix
\be
\ggrm_{rs} = \langle \xi_r | \xi_s \rangle
\ee
Clearly the Gram matrix cannot determine the POVM uniquely since if $U$ is any unitary operator then the vectors $U|\xi_r\rangle$ will define another POVM having the same Gram matrix.  However, the theorem we now prove shows that this is the only freedom.  In other words, the Gram matrix fixes the POVM up to unitary equivalence.  The theorem also provides us with  a criterion for deciding whether an arbitrary $d^2 \times d^2$ matrix $\ggrm$ is  the Gram matrix corresponding to a POVM of the specified type.  As a corollary this will give us a criterion for deciding whether an arbitrary tensor $\theta_{rs}$ is specifically the order-$2$ angle tensor for a SIC-family.

\begin{theorem}
\label{thm:GramSICcondition}
Let $\ggrm$ be any  $d^2 \times d^2$ Hermitian matrix.  Then the following conditions are equivalent:
\begin{enumerate}
\item[(1)] $\ggrm$ is a rank $d$ projector.
\item[(2)] $\ggrm$ satisfies the trace identities
\be
\Tr(\ggrm) = \Tr(\ggrm^2) = \Tr(\ggrm^3) = \Tr(\ggrm^4) = d
\ee
\item[(3)] $\ggrm$ is the Gram matrix for a set of $d^2$  vectors $|\xi_r\rangle$ (not all normalized) such that $ |\xi_r\rangle \langle \xi_r|$ is a POVM:
\begin{align}
 \langle \xi_r | \xi_s \rangle & = \ggrm_{rs}
\\
\sum_{r=1}^{d^2}  |\xi_r \rangle \langle \xi_r | & =  I
\end{align}
\end{enumerate}

Suppose  $\ggrm$ satisfies these conditions. To  construct a  POVM corresponding to $\ggrm$  let the $d$ column vectors 
\be
\bmt \xi_{11} \\ \xi_{12} \\ \vdots \\ \xi_{1d^2}
\emt,
\bmt \xi_{21} \\ \xi_{22} \\ \vdots \\ \xi_{2d^2}
\emt,
\dots,
\bmt \xi_{d1} \\ \xi_{d2} \\ \vdots \\ \xi_{dd^2}
\emt
\ee
be any orthonormal basis for the subspace onto which $\ggrm$ projects.  Define
\be
|\xi^{\vphantom{*}}_r\rangle =\sum_{a=1}^{d} \xi^{*}_{ar}|a\rangle
\label{eq:thm1PsiVecDef}
\ee
where the vectors $|a\rangle$ are  any orthonormal basis for $\mathcal{H}_d$.  Then 
$\ggrm$ is the Gram matrix for the vectors $|\xi_1\rangle, \dots , |\xi_{d^2}\rangle$.  Moreover, the necessary and sufficient condition for any other  set of vectors $|\eta_1\rangle, \dots, |\eta_{d^2}\rangle$ to have Gram matrix $\ggrm$ is that there exist  a unitary operator $U$ such that
\be
|\eta_r \rangle = U |\xi_r \rangle
\ee
for all $r$.
\end{theorem}
\begin{proof}
We begin by showing that $(3) \implies (1)$.
Suppose  $|\xi_1\rangle, \dots |\xi_{d^2}\rangle$ is any set of $d^2$ vectors such that  $|\xi_r\rangle \langle \xi_r | $ is  a POVM.  So
\be
\sum_{r=1}^{d^2} |\xi_r \rangle \langle \xi_r |  = I
\label{eq:thm1InterC}
\ee
Let
\be
\ggrm_{rs} =  \langle \xi_r | \xi_s \rangle
\ee
be the Gram matrix. 
Then $\ggrm$ is Hermitian. Moreover, $P^2 = P$ since 
\begin{align}
\sum_{t=1}^{d^2} \ggrm_{rt}\ggrm_{ts}& =  \langle \xi_r | \left( \sum_{t=1}^{d^2} |\xi_t \rangle \langle \xi_s | 
\right) |\xi_r \rangle 
\nonumber
\\
& =  \langle \xi_r | \xi_s \rangle
\nonumber
\\
& = \ggrm_{rs}
\end{align}
Also
\be
\Tr(\ggrm) = \sum_{r=1}^{d^2}  \langle \xi_r | \xi_r \rangle = d
\ee
(as can be seen by  taking the trace on both sides of Eq.~(\ref{eq:thm1InterC})).  
So $\ggrm$ is a rank-$d$ projector.

We next show that $(1) \implies (3)$.  Let $\ggrm$ be a rank-$d$ projector, and let the $d$ column vectors  
\be
\bmt \xi_{11} \\ \xi_{12} \\ \vdots \\ \xi_{1d^2}
\emt,
\bmt \xi_{21} \\ \xi_{22} \\ \vdots \\ \xi_{2d^2}
\emt,
\dots,
\bmt \xi_{d1} \\ \xi_{d2} \\ \vdots \\ \xi_{dd^2}
\emt
\ee
be an orthonormal basis for the subspace onto which it projects.  So
\begin{align}
\sum_{r=1}^{d^2} \xi^{*}_{ar} \xi^{\vpu{*}}_{br} &= \delta_{ab}
\label{eq:thm1InterA}
\\
\intertext{for all $a,b$, and}
\sum_{a=1}^{d^2} \xi^{\vpu{*}}_{ar} \xi^{*}_{as} & = \ggrm_{rs} 
\label{eq:thm1InterB}
\end{align}
for all $r,s$.  
Now let $|\xi_1\rangle, \dots |\xi_{d^2}\rangle$ be the vectors defined by Eq.~(\ref{eq:thm1PsiVecDef}).  Then it follows from Eq.~(\ref{eq:thm1InterA})  that 
\begin{align}
\sum_{r=1}^{d^2} |\xi_r \rangle \langle \xi_r | & =  \sum_{a,b=1}^{d} \left(\sum_{r=1}^{d^2} \xi^{*}_{ar} \xi^{\vpu{*}}_{br}\right) |a\rangle \langle b|
\nonumber
\\
& =  \sum_{a=1}^{d} |a\rangle \langle a |
\nonumber
\\
& =  I
\\
\intertext{implying that $ |\xi_r\rangle \langle \xi_r | $ is POVM.  Also, it follows from Eq.~(\ref{eq:thm1InterB}) that}
\langle \xi^{\vpu{*}}_r | \xi^{\vpu{*}}_s \rangle & = \sum_{a=1}^{d} \xi^{\vpu{*}}_{ar} \xi^{*}_{as}= \ggrm^{\vpu{*}}_{rs}
\end{align}
implying that the $|\xi_r\rangle$ have Gram matrix $\ggrm$.

We next turn to condition (2).  The fact that $(1) \implies (2)$ is immediate.  To prove the reverse implication observe that condition $(2)$ implies
\be
\Tr(\ggrm^4) - 2 \Tr (\ggrm^3) + \Tr(\ggrm^2) = 0
\label{eq:thm1InterD}
\ee
Let $\lambda_1, \dots, \lambda_{d^2}$ be the eigenvalues of $\ggrm$.  Then  Eq.~(\ref{eq:thm1InterD}) implies
\be
\sum_{r=1}^{d^2} \lambda_r^2 (\lambda_r-1)^2 = 0
\ee
It follows that each eigenvalue is either $0$ or $1$.  Since $\Tr(\ggrm)=d$ we must have $d$ eigenvalues $=1$ and the rest all $=0$.  So $\ggrm$ is a rank-$d$ projector.  

It remains to show that the POVM corresponding to a given rank-$d$ projector is unique up to unitary equivalence. To prove this let $\ggrm$ be a rank-$d$ projector,  let $|\xi_r\rangle$ be the vectors defined by  Eq.~(\ref{eq:thm1PsiVecDef}), and let  $|\eta_1\rangle, \dots , |\eta_{d^2}\rangle$ be any other set of vectors such that
\be
\langle \eta_r| \eta_s \rangle = P_{rs}
\ee
for all $r,s$.  Define
\be
\eta_{ar} = \langle \eta_r | a\rangle
\ee
Then
\begin{align}
\sum_{r=1}^{d^2} \eta^{*}_{ar} \eta^{\vpu{*}}_{br}  & = \langle a | \left(\sum_{r=1}^{d^2} |\eta_r \rangle \langle \eta_r | 
\right) |b\rangle
=\delta_{ab}
\\
\intertext{(because $|\eta_r\rangle \langle \eta_r |$ is a POVM) and}
\sum_{a=1}^{d} \eta^{\vpu{*}}_{ar} \eta^{*}_{as} & = \ggrm_{rs}
\end{align}
(because the $|\eta_r\rangle$ have Gram matrix $P$).  So the $d$ column vectors
\be
\bmt \eta_{11} \\ \eta_{12} \\ \vdots \\ \eta_{1d^2}
\emt,
\bmt \eta_{21} \\ \eta_{22} \\ \vdots \\ \eta_{2d^2}
\emt,
\dots,
\bmt \eta_{d1} \\ \eta_{d2} \\ \vdots \\ \eta_{dd^2}
\emt
\ee
are an orthonormal basis for the subspace onto which $\ggrm$ projects.  But the column vectors
\be
\bmt \xi_{11} \\ \xi_{12} \\ \vdots \\ \xi_{1d^2}
\emt,
\bmt \xi_{21} \\ \xi_{22} \\ \vdots \\ \xi_{2d^2}
\emt,
\dots,
\bmt \xi_{d1} \\ \xi_{d2} \\ \vdots \\ \xi_{dd^2}
\emt
\ee
are also an orthonormal basis for this subspace.  So there must exist a $d\times d$ unitary matrix $U_{ab}$ such that
\be
\eta_{ar} = \sum_{b=1}^{d} U_{ab} \xi_{br}
\ee
for all $a, r$.  Define
\be
U= \sum_{a,b=1}^{d} U^{*}_{ab} |a \rangle \langle b|
\ee
Then
\be
|\eta_r \rangle = U | \xi_r \rangle
\ee
for all $r$.
\end{proof}
In the case of a SIC-POVM we have 
\be
|\xi_r\rangle = \frac{1}{\sqrt{d}} |\psi_r\rangle
\ee
where the vectors $|\psi_r\rangle$ are normalized, and
\be
\ggrm_{rs} = \frac{1}{d} G_{rs} = \frac{1}{d} K_{rs} e^{i\theta_{rs}}
\label{eq:SICGramProjectorDef}
\ee
where $G$ is the Gram matrix of the vectors $|\psi_r\rangle$ and $\theta_{rs}$ is the order-$2$ angle tensor.  In the sequel we will distinguish these matrices by referring to $G$ as the Gram matrix and $\ggrm$ as the Gram projector.

We have
\begin{corollary}
\label{cor:order2AngleTensor}
Let $\theta_{rs}$ be a real anti-symmetric tensor.  Then the following statements are equivalent:
\begin{enumerate}
\item $\theta_{rs}$ is an order-$2$ angle tensor corresponding to a SIC-family.
\item $\theta_{rs}$ satisfies
\be
\sum_{t=1}^{d^2} K_{rt}K_{ts} e^{i(\theta_{rt}+\theta_{ts})} = d K_{rs} e^{i\theta_{rs}}
\label{eq:sicProjCond}
\ee
for all $r,s$.
\item $\theta_{rs}$ satisfies  
\begin{align}
\sum_{r,s,t=1}^{d^2} K_{rs}K_{st}K_{tr} e^{i (\theta_{rs} + \theta_{st} + \theta_{tr})} & = d^4
\label{eq:sicTraceCondA}
\\
\intertext{and}
\sum_{r,s,t,u=1}^{d^2} K_{rs} K_{st}K_{tu}K_{ur} e^{i(\theta_{rs}+\theta_{st} + \theta_{tu} + \theta_{ur})} & = d^5
\label{eq:sicTraceCondB}
\end{align}
\end{enumerate} 

Let $\Pi\vpu{'}_r$,  $\Pi'_r$ be two different SIC-sets, and let $\theta\vpu{'}_{rs}$, $\theta'_{rs}$ be  corresponding  order-$2$ angle tensors.  Then there exists a unitary $U$ such that
\be
\Pi'_r = U \Pi\vpu{'}_r U^{\dagger}
\ee
for all $r$    if and only if 
\be
\theta'_{rs} = \theta\vpu{'}_{rs} - \phi_r + \phi_s 
\ee
for some arbitrary set of phase angles $\phi_r$ (in other words two SIC-sets are unitarily equivalent if and only if their order-$2$ angle tensors are gauge equivalent).  

A SIC-family can be reconstructed from its order-$2$ angle tensor $\theta_{rs}$ by calculating an orthonormal basis for the subspace onto which the Gram projector
\be
\ggrm_{rs} = \frac{1}{d} K_{rs} e^{i\theta_{rs}}
\ee
projects, as described in Theorem~~\ref{thm:GramSICcondition}.
\end{corollary}
\begin{remark}
The sense in which we are using the term ``gauge equivalence'' is explained   in the passage immediately following Eq.~(\ref{eq:KrsDef}).

Note that condition $(2)$ imposes $d^2(d^2-1)/2$ independent constraints  (taking account of the anti-symmetry of $\theta_{rs}$).  Condition $(3)$, by contrast, only imposes $2$ independent constraints.  It is to be observed, however, that the price we pay for the reduction in the number of equations is that Eqs.~(\ref{eq:sicTraceCondA}) and~(\ref{eq:sicTraceCondA}) are respectively cubic and quartic in the phases, whereas Eq.~(\ref{eq:sicProjCond}) is only quadratic.
\end{remark}
\begin{proof}
Let $\theta_{rs}$ be an arbitrary anti-symmetric tensor, and define
\be
\ggrm_{rs} = \frac{1}{d} K_{rs} e^{i\theta_{rs}}
\ee
The anti-symmetry of $\theta_{rs}$ means that $\ggrm$ is automatically Hermitian.  So it follows from Theorem~\ref{thm:GramSICcondition} that a necessary and sufficient condition for $P_{rs}$ to be a rank-$d$ projector, and for $\theta_{rs}$ to be the order-$2$ angle tensor of a SIC-family, is that
\be
\sum_{t=1}^{d^2} K_{rt}K_{ts} e^{i(\theta_{rt}+\theta_{ts})} = d K_{rs} e^{i\theta_{rs}}
\ee
for all $r,s$. 

To prove the equivalence of conditions $(1)$ and $(3)$ note that the conditions $\Tr(\ggrm) = \Tr(\ggrm^2)=d$ are an automatic consequence of $\ggrm$ having the specified form.  So it follows from Theorem~\ref{thm:GramSICcondition} that $\theta_{rs}$ is the order-$2$ angle tensor of a SIC-family if and only if Eqs.~(\ref{eq:sicTraceCondA}) and~(\ref{eq:sicTraceCondB}) are satisfied. 

Now let $\Pi\vpu{'}_r$, $\Pi'_r$ be two SIC-sets and let $\theta\vpu{'}_{rs}$, $\theta'_{rs}$ be  order-$2$ angle tensors corresponding to them.  Then there exist normalized vectors $|\psi\vpu{'}_r\rangle$, $|\psi'_{r}\rangle$ such that
\begin{align}
\Pi\vpu{'}_r & = |\psi\vpu{'}_r\rangle \langle \psi\vpu{'}_r |   &  \Pi'_r & = |\psi'_r \rangle \langle \psi'_r | 
\\
\intertext{for all $r$, and}
\langle \psi\vpu{'}_r | \psi\vpu{'}_s \rangle & = K_{rs} e^{i \theta_{rs}}  & \langle \psi'_r | \psi'_s \rangle & = K_{rs} e^{i\theta'_{rs}}
\end{align}
for all $r,s$.  

Suppose, first of all, that there exists a unitary $U$ such that
\be
\Pi'_r = U \Pi\vpu{'}_r U^{\dagger}
\ee
Then there exist phase angles $\phi_r$ such that 
\be
|\psi'_r \rangle = e^{i\phi_r} U |\psi\vpu{'}_r \rangle
\ee
for all $r$, which  is easily seen to imply that 
\be
\theta'_{rs} = \theta_{rs} - \phi_r + \phi_s
\ee
for all $r,s$.  So $\theta\vpu{'}_{rs}$, $\theta'_{rs}$  are gauge equivalent.

Conversely, suppose  there exist phase angles $\phi_r$ such that
\be
\theta'_{rs} = \theta\vpu{'}_{rs} -\phi_r + \phi_s
\ee
Define
\be
|\psi''_{r}\rangle = e^{-i \phi_r} |\psi'_r\rangle
\ee
Then
\be
\langle \psi''_r | \psi''_s\rangle = K_{rs} e^{i\theta_{rs}} = \langle \psi\vpu{'}_r | \psi\vpu{'}_s\rangle
\ee
for all $r,s$.  So it follows from Theorem~\ref{thm:GramSICcondition} that there exists a unitary $U$ such that
\be
|\psi''_r \rangle = U |\psi\vpu{'}_r \rangle
\ee
for all $r$.  Consequently
\be
\Pi'_r = |\psi''_r \rangle \langle \psi''_r | =  U \Pi\vpu{'}_r U^{\dagger}
\ee
for all $r$.  So $\Pi\vpu{'}_r$ and $\Pi'_r$ are unitarily equivalent. 
\end{proof}

We now turn to the order-$3$ angle tensors.  We have
\begin{theorem}
\label{thm:TripProdSICcondition}
Let $\theta_{rst}$ be a real completely anti-symmetric tensor.  Then the following conditions are equivalent:
\begin{enumerate}
\item $\theta_{rst}$ is the order-$3$ angle tensor for a  SIC-family
\item For some fixed $a$ and all $r,s,t$
\begin{align}
\theta_{a
rs}+\theta_{ast} + \theta_{atr} & = \theta_{rst}
\label{eq:thm2SICcondA} 
\\
\intertext{and for all $r,s$}
\sum_{t=1}^{d^2} K_{rt} K_{ts} e^{i\theta_{rst}} &= d K_{rs}
\label{eq:thm2SICcondB}  
\end{align}
\item For some fixed $a$ and all $r,s,t$
\begin{align}
\theta_{ars}+\theta_{ast} + \theta_{atr} & = \theta_{rst} 
\label{eq:thm2SICcondC} 
\\
\intertext{and}
\sum_{r,s,t=1}^{d^2} K_{rs}K_{st}K_{tr} e^{i\theta_{rst}} & = d^4
\label{eq:thm2SICcondD} 
\\
\sum_{r,s,t,u=1}^{d^2} K_{rs} K_{st}K_{tu} K_{ur} e^{i(\theta_{rst} + \theta_{tur})} & = d^5
\label{eq:thm2SICcondE} 
\end{align}
\end{enumerate}

Let $\Pi\vpu{'}_r$, $\Pi'_r$ be two different SIC-sets and let $\theta\vpu{'}_{rst}$, $\theta'_{rst}$ be the corresponding order-$3$ angle tensors.  Then the necessary and sufficient condition for there to exist a unitary $U$ such that
\be
\Pi'_r = U \Pi\vpu{'}_r U^{\dagger}
\ee 
for all $r$ is that $\theta'_{rst} = \theta\vpu{'}_{rst}$ for all $r,s,t$ (in other words two SIC-sets are unitarily equivalent if and only if their order-$3$ angle tensors are identical).

Let $\theta_{rst}$ be the order-$3$ angle tensor corresponding to a SIC-family. Then the order-$2$ angle tensor is given by (up to gauge freedom)
\be
\theta_{rs} = \theta_{ars}
\ee
for any fixed $a$, from which the SIC-family can be reconstructed using the method described in Theorem~\ref{thm:GramSICcondition}.
\end{theorem}
\begin{remark}
Unlike the order-$2$ tensor, the order-$3$ angle tensor is gauge invariant.  This means that it provides what is, in many ways, a more useful characterization of the SIC-family.  For that reason we will be almost exclusively concerned with the order-$3$ tensor in the remainder of this paper.
\end{remark}
\begin{proof}
The fact  that $(1)\implies (2)$ is an immediate consequence of the definition of the order-$3$ angle tensor and condition (2) of Corollary~\ref{cor:order2AngleTensor}.  To prove that $(2)\implies (1)$  let $\theta_{rst}$ be a completely anti-symmetric tensor such that condition $(2)$ holds.  Define
\be
\theta_{rs} = \theta_{ars}
\ee
for all $r,s$.  Then Eq.~(\ref{eq:thm2SICcondB}) implies
\be
\sum_{t=1}^{d^2} K_{rt} K_{ts} e^{i(\theta_{rt}+ \theta_{ts})}
 = e^{i\theta_{rs}}\left(\sum_{t=1}^{d^2} K_{rt} K_{ts} e^{i\theta_{rst}} \right)^{*}= d K_{rs} e^{i\theta_{rs}}
\ee
for all $r,s$.    It follows from Corollary~\ref{cor:order2AngleTensor} that $\theta_{rs}$ is the order-$2$ and $\theta_{rst}$  the order-$3$ angle tensor of a SIC-family.

The equivalence of conditions $(1)$ and $(3)$ is proved similarly.  

It remains to show that two SIC-sets are unitarily equivalent if and only if their order-$3$ angle tensors are identical.
To see this let $\Pi\vpu{'}_r = |\psi\vpu{'}_r\rangle \langle \psi\vpu{'}_r|$ and $\Pi'_r = |\psi'_r \rangle \langle \psi'_r |$ be two different SIC-sets having the same order-$3$ angle tensor $\theta_{rst}$.  Let $\theta\vpu{'}_{rs}$ (respectively $\theta'_{rs}$) be the order-$2$ angle tensor corresponding to the vectors $|\psi\vpu{'}_r\rangle$ (respectively $|\psi'_r\rangle$).  Choose some fixed index $a$.  We have
\begin{align}
\theta'_{ar}+\theta'_{sa} + \theta'_{rs} & =\theta\vpu{'}_{ar} + \theta\vpu{'}_{sa} + \theta\vpu{'}_{rs}
\\
\intertext{for all $r,s$.  Consequently}
\theta'_{rs} & = \theta_{rs}+\phi_{r}-\phi_{s}
\\
\intertext{for all $r,s$, where}
\phi\vpu{'}_{r} & = \theta\vpu{'}_{ar} - \theta'_{ar}  
\end{align}
So $\theta'_{rs}$ and $\theta\vpu{'}_{rs}$ are gauge equivalent.  It follows from Corollary~\ref{cor:order2AngleTensor} that $\Pi\vpu{'}_r$ and $\Pi'_r$ are unitarily equivalent.  Conversely, suppose that $\Pi\vpu{'}_r$ and $\Pi'_r$ are unitarily equivalent, and let $\theta\vpu{'}_{rs}$, $\theta'_{rs}$ be  order-$2$ angle tensors corresponding to them.     It follows from Corollary~\ref{cor:order2AngleTensor}  that $\theta\vpu{'}_{rs}$ and $\theta'_{rs}$ are gauge equivalent.  It is then immediate that the order-$3$ angle tensors are identical.
\end{proof}
Finally, let us note that when expressed in terms of the triple products  Eq.~(\ref{eq:thm2SICcondB}) reads
\begin{align}
\sum_{t=1}^{d^2} T^{\vpu{2}}_{rst} & = d K^2_{rs}
\label{eq:sec2FinalA}
\\
\intertext{while Eq.~(\ref{eq:thm2SICcondD}) reads}
\sum_{r,s,t=1}^{d^2} T_{rst} & = d^4
\label{eq:sec2FinalB}
\\
\intertext{For Eq.~(\ref{eq:thm2SICcondE}) we have to work a little harder.  We have}
\sum_{r,s,t,u=1}^{d^2} \frac{1}{K^2_{rt}} T_{rst}T_{tur} &= d^5
\end{align}
from which it follows
\begin{align}
 d^5 & = \sum_{r,s,t,u=1}^{d^2} \bigl( -d\delta_{rt} + d+1\bigr)T_{rst}T_{tur} 
\nonumber
\\
& =(d+1)\sum_{r,s,t,u=1}^{d^2} T_{rst}T_{tur}  -d \sum_{r,s,u=1}^{d^2}  K^2_{rs}K^2_{ru} 
\nonumber
\\
& = (d+1)\sum_{r,s,t,u=1}^{d^2} T_{rst}T_{tur} - d^5
\end{align}
Consequently
\be
\sum_{s,u=1}^{d^2} \Tr\bigl(T_s T_u\bigr) = \sum_{r,s,t,u=1}^{d^2} T_{rst}T_{tur} = \frac{2d^5}{d+1}
\label{eq:sec2FinalC}
\ee
 This equation be alternatively written
\be
\sum_{r,s=1}^{d^2} \Tr\bigl(T_r T_s\bigr) = \frac{2d^5}{d+1}
\ee
where  $T_r$ is the matrix with matrix elements $(T_r)_{uv} = T_{ruv}$.

When they are written like this, in terms of the triple products, the fact that Eq.~(\ref{eq:sec2FinalA}) implies Eqs.~(\ref{eq:sec2FinalB}) and~(\ref{eq:sec2FinalC}) becomes almost obvious.  The reverse implication, by contrast, is rather less obvious.

\section{Spectral Decompositions}
\label{sec:QQTProp1}

Let $T_r$, $J_r$, $R_r$ be the $d^2\times d^2$ matrices whose matrix elements are
\begin{align}
(T_r)_{st} & = T_{rst} & (J_r)_{st} & = J_{rst} & (R_r)_{st} & = R_{rst} 
\end{align}
where $J_{rst}$, $R_{rst}$ are the quantities defined by Eqs.~(\ref{eq:JmatrixDef}) and~(\ref{eq:RmatrixDef}).
So $J_r$ is the adjoint representation matrix of $\Pi_r$. 
In this section we derive the spectral decompositions of  these matrices.  To avoid  confusion we will use the notation $|\psi\rangle$ to denote a ket in $d$ dimensional Hilbert space  $\mathcal{H}_d$, and  $\| \psi \dr$ to denote a ket in $d^2$ dimensional Hilbert space $\mathcal{H}_{d^2}$. In terms of this notation the spectral decompositions will turn out to be:  
\begin{align}
T_r & = \frac{d}{d+1} Q_r + \frac{2d}{d+1} \| e_r \dr \dl e_r \| 
\label{eq:TrSpecDecomp}
\\
J^{\vpu{T}}_r & = Q^{\vpu{T}}_r - Q^{\mathrm{T}}_r
\label{eq:JtermsQQT}
\\
R^{\vpu{T}}_r & = Q^{\vpu{T}}_r + Q^{\mathrm{T}}_r + 4 \| e_r \dr \dl e_r \|
\label{eq:RtermsQ}
\end{align}
In these expressions the vector $\| e_r \dr$ is normalized, and its components in the standard basis are all real.    $Q_r$ is a rank $d-1$ projector such that
\be
Q^{\vpu{T}}_r \| e_r \dr = Q^{\mathrm{T}}_r\| e_r \dr = 0
\ee
and which has, in addition, the remarkable property of being orthogonal to its own transpose (also a rank $d-1$ projector):
\be
Q^{\vpu{T}}_r Q^{\mathrm{T}}_r = 0
\ee
Explicit expressions for $\| e_r \dr$ and $Q_r$ will be given below.

It will be convenient to define the rank $2(d-1)$ projector
\begin{align}
\bar{R}_r &= Q^{\vphantom{\mathrm{T}}}_r + Q^{\mathrm{T}}_r
\label{eq:RbarDef}
\\
\intertext{We  have}
\bar{R}_r & = J^2_r
\\
\intertext{and}
R_r & = \bar{R}_r + 4 \| e_r \dr \dl e_r \| 
\end{align}
Since $Q_r$ is Hermitian we have 
\be
Q^{\mathrm{T}}_r = Q^{*}_r
\ee
where $Q^{*}_r$ is the matrix whose elements are the complex conjugates of the corresponding elements of $Q_r$.  So $\bar{R}_r$ is twice the real part of $Q_r$ and $- iJ_r$ is twice the imaginary part. 
 
In Section~\ref{sec:QQTprop2} we will show that Eq.~(\ref{eq:JtermsQQT}) is essentially definitive of a SIC-POVM.  To be more specific, let $\lie_r$ be any set of $d^2$ Hermitian matrices which constitute a basis for $\gl(d,\mathbb{C})$, and let $\str_r$ be the adjoint representative of $\lie_r$ in that basis.  Then we will show that $\str_r$ has the spectral decomposition
\be
\str^{\vphantom{\mathrm{T}}}_r = Q^{\vphantom{\mathrm{T}}}_{r} - Q^{\mathrm{T}}_{r}
\ee
where $Q_r$ is a rank $d-1$ projector which is orthogonal to its own transpose if and only if the $\lie_r$ are a family of SIC projectors up to multiplication by a sign and shifting by a  multiple of the identity.

Having stated our results let us now turn to the task of proving them.  We begin by deriving the spectral decomposition of $T_r$.  Multiplying both sides of the equation   
\be
\Pi_r \Pi_s = \frac{d+1}{d} \sum_{t=1}^{d^2} T_{rst} \Pi_t - \mt^{2}_{rs} I
\label{eq:QQTnecessityProd}
\ee
by $\Pi_r$ we find
\begin{align}
\Pi_r \Pi_s & = \frac{d+1}{d} \sum_{t=1}^{d^2} T_{rst}\Pi_r\Pi_t - \mt^2_{rs} \Pi_r
\nonumber
\\
& = \frac{(d+1)^2}{d^2} \sum_{t=1}^{d^2} (T_r)^2_{st} \Pi_t - \frac{d+1}{d}\sum_{t=1}^{d^2} T_{rst} \mt^2_{rt} I - \mt^2_{rs} \Pi_r
\end{align}
We have
\begin{align}
\sum_{t=1}^{d^2} T_{rst} \mt^2_{rt}
& = \frac{1}{d+1} \sum_{t=1}^{d^2} T_{rst} (d\delta_{rt} +1)
\nonumber
\\
& = \frac{1}{d+1}\left( d T_{rsr} +  \sum_{t=1}^{d^2} T_{rst}\right)
\nonumber
\\
& = \frac{2d}{d+1} T_{srr}
\nonumber
\\
& =\frac{2d}{d+1} \mt^2_{rs}
\label{eq:QQT1TKsumrel}
\end{align}
Consequently
\be
\Pi_r \Pi_s  =
\frac{d+1}{d}\sum_{t=1}^{d^2} \left(\frac{d+1}{d} (T_r)^2_{st}- \mt^2_{rs}\mt^{2}_{rt}
\right) \Pi_t - \mt^2_{rs} I
\ee
Comparing with Eq.~(\ref{eq:QQTnecessityProd}) we deduce
\be
(T_r)^2_{rs}  = \frac{d}{d+1} T_{rst} + \frac{d}{d+1} \mt^2_{rs}\mt^2_{rt}
\label{eq:QQT1T2form}
\ee
Now define
\be
\|e_r\dr = \sqrt{\frac{d+1}{2d}}\sum_{s=1}^{d^2} \mt^2_{rs} \| s\dr
\label{eq:erVecDef}
\ee
where the basis kets $\|s\dr$ are given by (in column vector form)
\be
\|1\dr = \bmt 1 \\ 0 \\  \vdots \\ 0\emt,
\|2\dr = \bmt 0 \\ 1 \\  \vdots \\ 0\emt,
\dots,
\|d^2\dr = \bmt 0 \\ 0 \\ \vdots \\ 1\emt
\label{eq:hd2StandardBasis}
\ee
It is easily verified that $\|e_r\dr$ is normalized.  Eq.~(\ref{eq:QQT1T2form}) then becomes
\be
T^2_r = \frac{d}{d+1}T^{\vphantom{2}}_r + \frac{2d^2}{(d+1)^2} \| e_r\dr \dl e_r \|
\label{eq:TsqRel}
\ee
Using  Eq.~(\ref{eq:QQT1TKsumrel}) we find
\begin{align}
\dl s \| T_r \| e_r \dr & = 
\sqrt{\frac{d+1}{2d}}\sum_{t=1}^{d^2} T_{rst} \mt^2_{rt}
\nonumber
\\
& = \sqrt{\frac{2d}{d+1}} \mt^2_{rs}
\nonumber
\\
& = \frac{2d}{d+1} \dl s \|  e_r \dr
\label{eq:erEvecOfQr}
\end{align}
So $\|e_r\dr$ is an eigenvector of $T_r$ with eigenvalue $\frac{2d}{d+1}$.  

Also define
\begin{align}
Q_r &= \frac{d+1}{d}T_r -2  \|e_r\dr \dl e_r \|
\label{eq:Qdef}
\\
\intertext{So in terms of the order-$3$ angle tensor the matrix elements of $Q_r$ are}
 Q_{rst} &= \frac{d+1}{d}\mt_{rs}\mt_{rt} \left(\mt_{st} e^{i\theta_{rst}} -  \mt_{rs}\mt_{rt}\right)
\label{eq:Qexplicit}
\end{align}
$Q_r$ is Hermitian (because $T_r$ is Hermitian).  Moreover
\be
Q^2_r  = \frac{(d+1)^2}{d^2} T^2_r -4 \|e_r\dr \dl e_r\| = Q_r
\ee
So $Q_r$ is a projection operator.  Since
\be
\Tr(T_r) = \sum_{u} T_{ruu} = \sum_{u=1}^{d^2} \mt^2_{ru} = d
\ee
we have
\be
\Tr(Q_r) = d-1
\ee
We have thus proved that the spectral decomposition of $T_r$ is
\be
T_r = \frac{d}{d+1} Q_r + \frac{2d}{d+1} \|e_r\dr \dl e_r \|
\ee
where $Q_r$ is a rank $d-1$ projector, as claimed. 

We next prove that $Q^{\mathrm{T}}_r\| e_r\dr = 0$. The fact that the components of $\|e_r\dr$ in the standard basis are all real means
\be
\dl s \| T^{\mathrm{T}}_r \| e_r\dr 
= \dl e_r \| T_r \| s\dr
=\frac{2d}{d+1} \dl s \|e_r\dr
\label{eq:erEvecTtrans}
\ee
So $\|e_r\dr$ is an eigenvector of $T^{\mathrm{T}}_r$ as well as $T^{\vphantom{\mathrm{T}}}_r$, again with the eigenvalue $\frac{2d}{d+1}$.  In view of Eq.~(\ref{eq:Qdef}) it follows that $Q^{\mathrm{T}}_r \| e_r \dr = 0$.

Turning to the problem of showing that $Q_r$ is orthogonal to its own transpose.  We have 
\begin{align}
Q^{\vphantom{\mathrm{T}}}_rQ^{\mathrm{T}}_r
& = 
\left(\frac{d+1}{d} T^{\vphantom{\mathrm{T}}}_r - 2 \|e_r\dr \dl e_r \|
\right)
\left(\frac{d+1}{d} T^{\mathrm{T}}_r - 2 \|e_r\dr \dl e_r \|
\right)
\nonumber
\\
&= \frac{(d+1)^2}{d^2} T^{\vphantom{\mathrm{T}}}_r T^{\mathrm{T}}_r - 4 \| e_r\dr \dl e_r \|
\end{align}
It follows from Eq.~(\ref{eq:TripProdDefB}) that
\begin{align}
\dl s \| T^{\vphantom{\mathrm{T}}}_r T^{\mathrm{T}}_{r} \|t\dr
& = \sum_{u=1}^{d^2} T_{rsu} T_{rtu} 
\nonumber
\\
& = G_{rs}G_{rt}\sum_{u=1}^{d^2} G_{su}G_{tu} G_{ur}G_{ur}
\\
\intertext{In view of Eq.~(\ref{eq:2designProp}) (\emph{i.e.}\ the fact that every SIC-POVM is a $2$-design) this implies}
\dl s \| T^{\vphantom{\mathrm{T}}}_r T^{\mathrm{T}}_{r} \|t\dr
& = \frac{2d}{d+1} |G_{rs}|^2|G_{rt}|^2
\nonumber
\\
& = \frac{2d}{d+1} \mt^2_{rs} \mt^2_{rt}
\nonumber
\\
& = \frac{4d^2}{(d+1)^2} \dl s \| \ev_r\dr \dl \ev_r \| t\dr
\\
\intertext{So}
T^{\vphantom{\mathrm{T}}}_r T^{\mathrm{T}}_{r}  &= \frac{4d^2}{(d+1)^2} \| \ev_r \dr \dl \ev_r \|
\\
\intertext{and consequently}
Q^{\vphantom{\mathrm{T}}}_r Q^{\mathrm{T}}_r & = 0
\end{align}  

Eqs.~(\ref{eq:JtermsQQT}) and~(\ref{eq:RtermsQ}) are immediate consequences of the results already proved and the definitions of $J_r$, $R_r$.

We defined the  $J$ matrices  to be the adjoint representatives of the SIC-projectors, considered as a basis for the Lie algebra $\gl (d,\mathbb{C})$, and that is certainly a most important fact about them. However, the results of this section show that, along with the vectors $\| e_r \dr$, they  actually determine the whole structure. Specifically, we have
\begin{align}
Q_r & = \frac{1}{2} \bigl( J_r + J^2_r\bigr)
\\
R_r & = J^2_r + 4 \|\ev_r \dr \dl \ev_r \| 
\\
T_r & = \frac{d}{2(d+1)} \Bigl( J_r + J^2_r + 4 \| \ev_r \dr \dl \ev_r \| \Bigr)
\end{align}
Moreover, if we know the $T$ matrices then we know the order-$3$ angle tensor, which in view of Theorem~\ref{thm:TripProdSICcondition} means we can reconstruct the SIC-projectors.  Since the vectors $\| \ev_r \dr$ are given, once and for all, this means that the problem of proving the existence of a SIC-POVM in dimension $d$ is  equivalent to the problem of proving the existence of a certain remarkable structure in the adjoint representation of $\gl(d,\mathbb{C})$ (as we will see in more detail in Section~\ref{sec:QQTprop2}).

In the Introduction we began with the concept of a SIC-POVM, and then defined the $J$ matrices in terms of it.  However,  one could, if one wished, go in the opposite direction, and  take the Lie algebraic structure to be primary, with the SIC-POVM being the secondary, derivative entity.

\section{The \texorpdfstring{$Q$-$Q^{\mathrm{T}}$}{Q-QT} Property}
\label{sec:QQTpropGen}

The next five sections are devoted to a study of the $J$ matrices which, as we will see, have numerous interesting properties. 
We begin our investigation  by trying to get some additional insight into what we will call the $Q$-$Q^{\mathrm{T}}$ property:  namely, the fact that the $J$ matrices have the spectral decomposition
\be
J^{\vpu{T}}_r = Q^{\vpu{T}}_r -Q^{\mathrm{T}}_r
\ee
where $Q_r$ is a rank $d-1$ projector which is orthogonal to its own transpose.  We wish to characterize the general class of matrices which are of this type. 
The following theorem provides one such characterization.
\begin{theorem}
\label{thm:QQstarPropA}
Let $A$ be a Hermitian matrix.  Then the following statements are equivalent:
\begin{enumerate}
\item $A$ has the spectral decomposition
\be
A = P - P^{\mathrm{T}}
\label{eq:thm4EqA}
\ee
where $P$ is a projector which is orthogonal to its own transpose.
\item $A$ is pure imaginary and $A^2$ is a projector.
\end{enumerate}
\end{theorem}
\begin{proof}
To show that $(1)\implies (2)$ observe that the fact that $P$ is Hermitian means 
\be
P^{\mathrm{T}} = P^{*}
\ee
where $P^{*}$ is the matrix whose elements are the complex conjugates of the corresponding elements of $P$.  So Eq.~(\ref{eq:thm4EqA}) implies that the components of $A$ are pure imaginary.  Since $PP^{\mathrm{T}}=0$   it also implies that $A^2$ is a projector.

To show that $(2) \implies (1)$ observe that the fact that $A^2$ is a projector means that the eigenvalues of $A$   $=\pm 1$ or $0$.  So
\be
A = P - P'
\ee
where $P$, $P'$ are orthogonal projectors.  Since $A$ is pure imaginary we must have
\be
P^{\mathrm{T}} - (P')^{\mathrm{T}} = A^{\mathrm{T}} = A^{*} = - A = P' - P
\ee
$P^{\mathrm{T}}$ and $(P')^{\mathrm{T}}$ are also orthogonal projectors.  So if  $P^{\mathrm{T}}|\psi\rangle = |\psi\rangle$, and $|\psi\rangle$ is normalized, we must have
\begin{align}
1 & = \langle \psi | P^{\mathrm{T}} |\psi\rangle 
\nonumber
\\
& = \big\langle \psi \bigl| \bigl(P^{\mathrm{T}} - (P')^{\mathrm{T}}\bigr) \bigr| \psi\big\rangle
\nonumber
\\
& = \langle\psi | P' |\psi \rangle - \langle\psi | P |\psi \rangle
\end{align}
Since
\begin{align}
0 &\le \langle\psi | P' |\psi \rangle \le 1
\\
0 &\le  \langle\psi | P |\psi \rangle \le 1
\end{align}
we must have $\langle\psi | P' |\psi \rangle=1$, implying $P'|\psi \rangle = |\psi \rangle$.  Similarly $P'|\psi\rangle = |\psi \rangle $ implies $P^{\mathrm{T}}|\psi\rangle = |\psi \rangle$.  So
\be
P' = P^{\mathrm{T}}
\ee
\end{proof}
We also have the following statement, inspired in part by Ref.~\cite{Goodson09},
\begin{theorem}
\label{thm:QQstarPropB}
The necessary and sufficient condition for a matrix $P$ to be a projector which is orthogonal to its own transpose is that 
\be
P = S D S^{\mathrm{T}}
\ee
where $S$ is an any real orthogonal matrix and $D$ has the block-diagonal form
\be
D = 
\bmt 
\sigma &  \dots & 0 & 0 &\dots & 0 \\
\vdots  & &\vdots & \vdots & & \vdots \\
0  &  \dots & \sigma & 0 & \dots & 0 \\
0 &  \dots & 0 & 0 & \dots & 0 \\
\vdots  & & \vdots & \vdots & & \vdots\\
0 & \dots & 0 & 0 & \dots & 0 
\emt           
\label{eq:thm5DmtDef}
\ee
with 
\be
\sigma = \frac{1}{2} \bmt 1 & - i \\ i & 1 \emt
\ee
In other words $D$ has $n$ copies of $\sigma$ on the diagonal, where $n = \rnk(P)$, and $0$ everywhere else.
\end{theorem}
\begin{proof}
Sufficiency is an immediate consequence of the fact that $\sigma$ is a rank $1$ projector such that $\sigma \sigma^{\mathrm{T}} = 0$.

To prove necessity let $d$ be the dimension of the space and $n$ the rank of $P$.  It will be convenient to define
\be
|1\rangle  = \bmt 1 \\ 0 \\ \vdots \\ 0 \emt , \hspace{0.2 in}  |2\rangle  = \bmt 0 \\ 1 \\ \vdots \\ 0 \emt , \hspace{0.2 in} \dots   \hspace{0.2 in}
|d\rangle  = \bmt 0 \\ 0 \\ \vdots \\ 1 \emt
\ee
In terms of these basis vectors we have
\be
P = \sum_{r,s=1}^{d} P_{rs} |r\rangle \langle s |
\ee
Now let $|\av_1\rangle , \dots, |\av_n\rangle$ be an orthonormal basis for the subspace onto which $P$ projects, and let $|\av^{*}_r\rangle$ be the column vector which is obtained from $|\av_r\rangle$ by taking the complex conjugate of each of its components.  Taking complex conjugates on each side of the equation
\begin{align}
P |\av_r\rangle & = |\av_r\rangle
\\
\intertext{gives}
P^{*} | \av^{*}_r \rangle & = |\av^{*}_r\rangle
\end{align}
So $|\av^{*}_1\rangle, \dots , |\av^{*}_n\rangle$ is an orthonormal basis for the subspace onto which $P^{\mathrm{T}} = P^{*}$ projects.  Since $P^{\mathrm{T}}$ is orthogonal to $P$ we conclude that
\be
\langle \av^{\vpu{*}}_r | \av^{*}_s \rangle = 0
\ee
for all $r,s$.

Next define  vectors $|\bv_1\rangle , \dots , |\bv_{2n}\rangle$ by
\begin{align}
|\bv_{2r-1}\rangle & = \frac{1}{\sqrt{2}} \bigl( |\av^{*}_r\rangle - |\av^{\vpu{*}}_r\rangle \bigr)
\\
|\bv_{2r} \rangle & = \frac{i}{\sqrt{2}} \bigl( |\av^{*}_r \rangle + | \av^{\vpu{*}}_r \rangle \bigr)
\end{align}
By construction these vectors are orthonormal and real.  So we can extend them to an orthonormal basis for the full space by adding a further $d-2n$ vectors $|\bv_{2n+1}\rangle, \dots, |\bv_d\rangle$, which can also be chosen to be real.  We have
\begin{align}
P &= \sum_{r=1}^{n} |\av_r\rangle \langle \av_r |
\nonumber
\\
& = \frac{1}{2}\sum_{r=1}^{n} \Bigl(|\bv_{2r-1}\rangle \langle \bv_{2r-1}| -i |\bv_{2r-1} \rangle \langle \bv_{2r}| + i |\bv_{2r} \rangle \langle \bv_{2r-1} | + |\bv_{2r} \rangle \langle \bv_{2r} |\Bigr)
\end{align}
So if we define
\begin{align}
S &= \sum_{r=1}^{d} |\bv_r \rangle \langle r | 
\\
\intertext{then $S$ is a real orthogonal matrix such that}
P &= S D S^{\mathrm{T}}
\end{align}
where
\be
D = \frac{1}{2} \sum_{r=1}^{n} \Bigl( |2r-1\rangle \langle 2r-1 | - i | 2r-1\rangle \langle 2r | + i | 2r \rangle \langle 2r-1| + |2r\rangle \langle 2r | \Bigr) 
\ee
is the matrix defined by Eq.~(\ref{eq:thm5DmtDef}).
\end{proof}
This result implies the following alternative characterization of the class of matrices to which the $J$ matrices belong
\begin{corollary}
Let $A$ be a Hermitian matrix.  Then the following statements are equivalent:
\begin{enumerate}
\item $A$ has the spectral decomposition
\be
A = P - P^{\mathrm{T}}
\label{eq:corr6EqA}
\ee
where $P$ is a projector which is orthogonal to its own transpose.
\item There exists a real orthogonal matrix $S$ such that
\be
A = S D S^{\mathrm{T}}
\ee 
where $D$ has the block diagonal form
\be
D = 
\bmt 
\sigma_y &  \dots & 0 & 0 &\dots & 0 \\
\vdots  & &\vdots & \vdots & & \vdots \\
0  &  \dots & \sigma_y & 0 & \dots & 0 \\
0 &  \dots & 0 & 0 & \dots & 0 \\
\vdots  & & \vdots & \vdots & & \vdots\\
0 & \dots & 0 & 0 & \dots & 0 
\emt           
\label{eq:corr6DmtDef}
\ee
$\sigma_y$ being the Pauli matrix
\be
\sigma_y = \bmt 0 & -i \\ i & 0 \emt
\ee
In other words $D$ has $n$ copies of $\sigma_y$ on the diagonal, where $n = \frac{1}{2}\rnk(A)$, and $0$ everywhere else (note that a matrix of this type must have even rank).
\end{enumerate}
\end{corollary}
\begin{proof}
Immediate consequence of Theorem~\ref{thm:QQstarPropB}.
\end{proof}

\section{Lie Algebraic Formulation of the Existence Problem}
\label{sec:QQTprop2}

This section is the core of the paper.  We show that the problem of proving the existence of a SIC-POVM in dimension $d$ is equivalent to the problem of proving the existence of an Hermitian basis for $\gl(d,\mathbb{C})$ all of whose elements have the $Q$-$Q^{\mathrm{T}}$ property.  We hope that this new way of thinking will help make the SIC-existence problem more amenable to solution.

The result we prove is the following:
\begin{theorem}
\label{thm:QQTEqualsSICness}
Let $\lie_r$ be a set of $d^2$ Hermitian matrices forming a basis for $\gl(d,\mathbb{C})$.  Let $\str_{rst}$ be the structure constants relative to this basis, so that
\be
[\lie_r, \lie_s] = \sum_{t=1}^{d^2} \str_{rst} \lie_t
\label{eq:thm7strConstDef}
\ee
and let $\str_r$ be the matrix with matrix elements $(\str_r)_{st} = \str_{rst}$.  Then the following statements are equivalent 
\begin{enumerate}
\item Each $\str_r$ has the spectral decomposition
\be
\str\vpu{T}_r = \liq\vpu{T}_r - \liq^{\mathrm{T}}_r
\label{eq:thm7QQTpropDef}
\ee
where $\liq_r$ is a rank $d-1$ projector which is orthogonal to its own transpose.
\item There exists a SIC-set $\Pi_r$, a set of signs $\sign_r = \pm 1$ and a real constant $\alpha \neq - \frac{1}{d}$ such that
\be
\lie_r = \sign_r (\Pi_r + \alpha I)
\ee
\end{enumerate}
\end{theorem}
\begin{remark}
The restriction to values of $\alpha \neq - \frac{1}{d}$ is needed to ensure that the matrices $\lie_r$ are linearly independent, and therefore constitute a basis for $\gl(d,\mathbb{C})$ (otherwise they would all have trace $=0$).  The $Q$-$Q^{\mathrm{T}}$ property continues to hold even if $\alpha$ does $= - \frac{1}{d}$.
\end{remark}

It will be seen that it is not only SIC-sets which have the $Q$-$Q^{\mathrm{T}}$ property, but also any set of operators obtained from a SIC-set by shifting by a constant and multiplying by an $r$-dependent sign.  So the $Q$-$Q^{\mathrm{T}}$ property is not strictly equivalent to the property of being a SIC-set.  However, it could be said that the properties are almost equivalent.  In particular, the existence of an Hermitian basis for $\gl(d,\mathbb{C})$ having the $Q$-$Q^{\mathrm{T}}$ property implies the existence of a SIC-POVM in dimension $d$, and conversely.  

\subsection*{Proof that (2)$\implies$(1)}
Taking the trace on both sides of 
\be
[\Pi_r , \Pi_s] =   \sum_{t=1}^{d^2} J_{rst} \Pi_t
\ee
we deduce that
\be
\sum_{t=1}^{d^2} J_{rst} = 0 
\label{eq:Jsum}
\ee
Then from the definition of $L_r$ in terms of $\Pi_r$ we find
\be
C_{rst} = \sign_r \sign_s \sign_t J_{rst}
\ee
Consequently
\be
C\vpu{T}_r = P\vpu{T}_r - P^{\mathrm{T}}_r
\ee
where
\be
P_r = \sign_r S Q_r S
\ee
 $S$ being the symmetric orthogonal matrix
\be
S = \bmt \sign_1 & 0 & \dots & 0 \\ 0 & \sign_2 & \dots & 0 \\ \vdots & \vdots &  & \vdots\\ 0 & 0 & \dots & \sign_{d^2} 
\emt
\ee
The claim is now immediate.
\subsection*{Proof that (1)$\implies$(2)}  For this we need to work harder.  Since the proof is rather lengthy we will break it into a number of lemmas.
 We first collect a few elementary facts which will be needed in the sequel:
\begin{lemma} 
\label{lem:elem}
Let $\lie_r$ be any Hermitian basis for $\gl(d,\mathbb{C})$, and let $\str_{rst}$ and $\str_r$ be the structure constants and adjoint representatives as defined in the statement of Theorem~\ref{thm:QQTEqualsSICness}.  Let $\liet_r=\Tr(L_r)$.  Then
\begin{enumerate}
\item The $\liet_r$ are not all zero.  
\item The $\str_{rst}$ are pure imaginary and antisymmetric in the first pair of indices.
\item The $\str_{rst}$ are completely antisymmetric if and only  if the $\str_r$ are Hermitian.
\item In every case
\be
\sum_{t=1}^{d^2} \str_{rst} \liet_t = 0
\label{eq:lem8A}
\ee
for all $r,s$.
\item In the special case that the $\str_r$ are Hermitian
\begin{align}
\sum_{r=1}^{d^2} \liet_r \lie_r &  = \cls I 
\\
\intertext{where}
\cls &= \frac{1}{d} \left( \sum_{r=1}^{d^2} \liet^2_r \right) > 0
\end{align}
\end{enumerate}
\end{lemma}
\begin{proof}
To prove (1) observe that if the $\liet_r$ were all zero it would mean that the identity was not in the span of the $\lie_r$---contrary to the assumption that they form a basis.

To prove  (2) observe that taking  Hermitian conjugates on both sides of Eq.~(\ref{eq:thm7strConstDef}) gives
\be
-[\lie_r , \lie_s] = \sum_{t=1}^{d^2} \str^{*}_{rst} \lie_t
\ee
from which it follows that $\str^{*}_{rst}= - \str\vpu{*}_{rst}$.  The fact that $\str\vpu{*}_{srt} = - \str\vpu{*}_{rst}$ is an immediate consequence of the definition.

(3) is now  immediate.  

(4) is proved in the same way as Eq.~(\ref{eq:Jsum}).

To prove  (5) observe that if the $\str_r$ are Hermitian it follows from (2) and (3) that  
\be
\sum_{r=1}^{d^2} \liet_r \str_{rst} = 0 
\ee
for all $s$, $t$.  Consequently the matrix
\be
\sum_{r=1}^{d^2} \liet_r \lie_r 
\ee
commutes with everything.  But the only matrices for which that is true are multiples of the identity.   It follows that 
\be
\sum_{r=1}^{d^2} \liet_r \lie_r = \cls I 
\ee 
for some real $\cls$.
Taking the trace on both sides of this equation we deduce
\be
\sum_{r=1}^{d^2} \liet^2_r = d \cls 
\ee  
The fact that $\cls > 0$ is a consequence of this and statement (1).  
\end{proof}
 We next observe that if the $\str_r$ have the $Q$-$Q^{\mathrm{T}}$ property they must, in particular, be Hermitian.  It turns out that that is, by itself, already  a very strong constraint.

Before stating the result it may be helpful if we explain the essential idea on which it depends.  Although we have not done so before, and will not do so again, it will be convenient to make use of  the covariant/contravariant  index notation which is often used to describe the structure constants.  Define the metric tensor
\be
\mtb_{rs} = \Tr(\lie_r\lie_s)
\ee
and let $\mtb^{rs}$ be its inverse.  So
\be
\sum_{t=1}^{d^2} \mtb^{rt}\mtb_{ts}
= 
\mtb^{r}_{\ph{r}s}
=
\begin{cases}
1 \qquad & r= s \\
0 \qquad & r \neq s
\end{cases}
\ee
We can use these tensors to raise and lower indices
(we use the Hilbert-Schmidt inner product for this purpose because the fact that $\gl(d,\mathbb{C})$ is not semi-simple means that its Killing form is degenerate~\cite{Lie1,Lie2,Lie3,Lie4}).  In particular, the matrices
\be
\lie^r = \sum_{t=1}^{d^2} \mtb^{rs} \lie_s
\ee
are the basis dual to the $\lie_r$:
\be
\Tr(\lie^r \lie_s ) = \mtb^{r}_{\ph{r} s} 
\ee
Suppose we now define structure constants $\stra_{rst}$ by
\be
[\lie_r, \lie_s] = \sum_{t=1}^{d^2} \stra_{rst} \lie^t
\label{eq:altStrCnstDef}
\ee
(so in terms of the $\str_{rst}$ we have $\stra_{rs}^{\ph{rs}t}=\str^{\vpu{t}}_{rst}$).
It follows from the relation
\be
\stra_{rst}  = \Tr\bigl( [\lie_r,\lie_s]\lie_t\bigr) = \Tr\bigl( \lie_r[\lie_s,\lie_t]\bigr) 
\ee
that the $\stra_{rst}$ are completely antisymmetric for any choice of the $\lie_r$.  If we now require that the matrices $\str_r$ be Hermitian it means that, not only the $\stra_{rst}$, but also the $\str_{rst}$ must be completely antisymmetric.  Since the two quantities are related by
\be
\stra_{rst} = \sum_{u=1}^{d^2} \str_{rsu} \mtb_{ut} 
\label{eq:ctwidTermsc}
\ee
this is a very strong requirement.  It means that the $\lie_r$ must, in a certain sense, be close to orthonormal (relative to the Hilbert-Schmidt inner product).  More precisely, it means we have the following lemma: 
\begin{lemma}
\label{lem:lem9ChermProp}
Let $\lie_r$,  $\str_{rst}$  and $\str_r$ be defined as in the statement of Theorem~\ref{thm:QQTEqualsSICness}, and let $\liet_r = \Tr(\lie_r)$.  Then the  $\str_r$ are Hermitian if and only if
\be
\Tr(\lie_r \lie_s) = \beta \delta_{rs} + \gamma \liet_r \liet_s
\label{eq:lem9HSinner}
\ee
where $\beta,\gamma$ are real constants such that $\beta> 0 $ and $\gamma < \frac{1}{d}$.  

If this condition is satisfied we also have
\begin{align}
\sum_{r=1}^{d^2} \liet_r \lie_r & = \frac{\beta}{1-d\gamma} I
\label{eq:lem9IdTermsLr}
\\
\sum_{r=1}^{d^2} \liet^2_r & = \frac{d\beta}{1-d\gamma}
\label{eq:lem9SqrdTraceSum}
\end{align}
\end{lemma}
\begin{proof}
To prove sufficiency observe that, in view of Eq.~(\ref{eq:ctwidTermsc}), the condition means
\be
\stra_{rst} = \beta \str_{rst} + \gamma \liet_t \sum_{u=1}^{d^2} \str_{rsu}\liet_u
\ee
In view of Lemma~\ref{lem:elem}, and the fact that $\beta\neq 0$,  this implies
\be
\str_{rst} = \frac{1}{\beta} \stra_{rst}
\ee
Since the $\stra_{rst}$ are completely antisymmetric we conclude that the $\str_{rst}$ must be also.  It follows that the $\str_r$ are Hermitian.

To prove necessity let $\stra_r$ (respectively $\mtb$) be the matrix whose matrix elements are $\stra_{rst}$ (respectively $\mtb_{st}$).  Then   Eq.~(\ref{eq:ctwidTermsc}) can be written
\begin{align}
\stra_r &= \str_r M
\\
\intertext{Taking the transpose (or, equivalently, the Hermitian conjugate) on both sides of this equation we find}
\stra_r& =  M \str_r 
\\
\intertext{implying }
[M,\str_r] & = 0
\\
\intertext{for all $r$.  Since the $\lie_r$ are a basis for $\gl(d,\mathbb{C})$ we deduce }
\left[M, \ad_A
\right]
&= 0
\end{align}
for all $A\in \gl(d,\mathbb{C})$. 
Eq.~(\ref{eq:lem9HSinner}) is a straightforward consequence of this, the fact that $\gl(d,\mathbb{C})$ has the direct sum decomposition $ \mathbb{C} I \oplus \sla(d,\mathbb{C})$, the fact that $\sla(d,\mathbb{C})$ is simple, and Schur's lemma~\cite{Lie1,Lie2,Lie3,Lie4}.  However, for the benefit of the reader who is not so familiar with the theory of Lie algebras we will give the argument in a little more detail.

Given  arbitrary  $A = \sum_{r=1}^{d^2} a_r L_r$, let $\| A\dr$ denote the column vector
\be
\|A\dr = \bmt a_1 \\ a_2 \\ \vdots \\ a_{d^2} \emt
\ee
So
\begin{align}
\| \lie_r \dr  &= \bmt 1 \\ 0 \\ \vdots \\ 0 \emt & 
\| \lie_2 \dr &= \bmt 0 \\ 1 \\ \vdots \\ 0 \emt &
\| \lie_{d^2} \dr & = \bmt 0 \\ 0 \\ \vdots \\ 1 \emt
\end{align}
In view of Lemma~\ref{lem:elem} we then have
\be
\| I \dr = \frac{1}{\cls} \sum_{r=1}^{d^2} \liet_r \|\lie_r\dr
\ee
Since
\be
\Tr(A) = \sum_{r=1}^{d^2} a_r \liet_r = \cls \dl I \| A\dr 
\label{eq:lem9TraceTermsI}
\ee
we have that $A\in \sla(d,\mathbb{C})$ if and only if $\dl I \| A\dr = 0$.

Now observe that it follows from Lemma~\ref{lem:elem} and the definition of $\mtb$ that 
\be
\mtb \| I \dr = \cls \| I \dr 
\ee
If $\mtb$ is a multiple of the identity we have $\mtb_{rs} = \cls \delta_{rs}$ and the lemma is proved.  Otherwise $\mtb$ has at least one more eigenvalue, $\beta$ say.  Let $\mathcal{E}$ be the corresponding eigenspace.  Since $\mathcal{E}$ is orthogonal to $\| I\dr$ it follows from Eq.~(\ref{eq:lem9TraceTermsI}) that $\mathcal{E} \subseteq \sla(d,\mathbb{C})$.  Since  $\mtb$ commutes with every adjoint representation matrix we have
\be
\ad_A \mathcal{E} \subseteq \mathcal{E}
\ee
for all $A\in \sla(d,\mathbb{C})$.  So $\mathcal{E}$ is an ideal of $\sla(d,\mathbb{C})$.  However $\sla(d,\mathbb{C})$ is a simple Lie algebra, meaning it has no proper ideals~\cite{Lie1,Lie2,Lie3,Lie4}.  So we must have $\mathcal{E} = \sla(d,\mathbb{C})$.  
  It follows that if we define
\begin{align}
\tilde{\lie}_r &= \lie_r - \frac{\liet_r}{d} I
\\
\intertext{then}
M \| \lie_r \dr & =\frac{\liet_r}{d}  M \| I \dr + M \|\tilde{\lie}_r\dr 
\\
& = \frac{\cls \liet_r}{d}\| I \dr + \beta \| \tilde{\lie}_r \dr
\\
& = \sum_{s=1}^{d^2} \left(\beta \delta_{rs} + \gamma \liet_r \liet_s\right) \|\lie_s\dr
\end{align}
where $\gamma = \frac{1}{d} \left( 1 - \frac{\beta}{\cls}\right)$.  Eqs.~(\ref{eq:lem9HSinner}), (\ref{eq:lem9IdTermsLr}) and~(\ref{eq:lem9SqrdTraceSum}) are now immediate (in view of Lemma~\ref{lem:elem}).

It remains to establish the bounds on $\beta, \gamma$.  
Let  $A = \sum_{r=1}^{d^2} a_r \lie_r$ be any non-zero element of $\sla(d,\mathbb{C})$.  Then $\sum_{r=1}^{d^2} a_r l_r =0$, so 
in view of Eq.~(\ref{eq:lem9HSinner}) we have
\be
0 < \Tr(A^2) = \beta \sum_{r=1}^{d^2} a^2_r 
\ee 
It follows that $\beta>0$.
Also, using Lemma~\ref{lem:elem} once more, we find
\begin{align}
\liet_r & = \frac{1}{\cls}\sum_{s=1}^{d^2} \liet_s \Tr(\lie_r\lie_s)
\nonumber
\\
& = \frac{\beta\liet_r}{\cls} + \frac{\gamma \liet_r}{\cls}\sum_{s=1}^{d^2} \liet^2_s
\nonumber
\\
& = \liet_r \left(\frac{\beta}{\cls} + d \gamma
\right)
\end{align} 
Since the $\liet_r$ cannot all be zero this implies
\be
\frac{\beta}{\cls} = 1 - d\gamma
\ee
Since $\frac{\beta}{\cls} > 0$ we deduce that $\gamma < \frac{1}{d}$.
\end{proof}
Eq.~(\ref{eq:lem9HSinner}) only depends on the $\str_r$ being Hermitian.  If we make   the assumption that the $\str_r$ have the $Q$-$Q^{\mathrm{T}}$ property we get a stronger statement:
\begin{corollary}
\label{corr:HSforQQTproperties}
Let $\lie_r$,  $\str_{rst}$  and $\str_r$ be as defined in the statement of Theorem~\ref{thm:QQTEqualsSICness}.  Suppose that the $\str_r$ have the spectral decomposition
\be
\str\vpu{T}_r = \liq\vpu{T}_r - \liq^{\mathrm{T}}_r
\ee
where $\liq_r$ is a rank $d-1$ projector which is orthogonal to its own transpose.
Then 
\begin{enumerate}
\item For all $r$
\be
\Tr(\lie_r) = \sign'_r \liet
\ee
\item For all $r,s$
\be
\Tr(\lie_r \lie_s) = \frac{d}{d+1} \delta_{rs} + \frac{\sign'_r\sign'_s}{d} \left( \liet^2 - \frac{1}{d+1} \right)
\label{eq:cor10HSInner}
\ee
\item 
\be
\sum_{r=1}^{d^2} \sign'_r \lie_r = d l I 
\label{eq:cor10Sum}
\ee
\end{enumerate}
for some real constant $\liet>0$ and signs $\sign'_r = \pm 1$. 
\end{corollary}
\begin{proof}
The proof relies on the fact that the Killing form for $\gl(d,\mathbb{C})$ is related to the Hilbert-Schmidt inner product by~\cite{Lie4}
\be
\Tr(\ad_A \ad_B) = 2 d \Tr(AB) - 2 \Tr(A) \Tr(B)
\ee
Specializing to the case $A = B = \lie_r$ and making use of the $Q$-$Q^{\mathrm{T}}$ property we find
\be
d-1 = d \Tr(\lie^2_r) - \liet^2_r
\ee
Using Lemma~\ref{lem:lem9ChermProp} we deduce
\be
\liet^2_r = \frac{d\beta-d+1}{1-d\gamma}
\label{eq:corr10lSqExp}
\ee
It follows that 
\be
\liet_r = \sign'_r \liet
\ee
for some real constant $\liet\ge 0$ and signs $\sign'_r = \pm 1$. The fact that the $\lie_r$ are a basis for $\gl(d,\mathbb{C})$ means the $\liet_r$ cannot all be zero.  So we must in fact have $\liet>0$.   Using this result in Eq.~(\ref{eq:lem9SqrdTraceSum}) we find
\begin{align}
\beta +d^2 \liet^2 \gamma &=  d \liet^2
\\
\intertext{while Eq.~(\ref{eq:corr10lSqExp}) implies}
d \beta + d \liet^2 \gamma & = d-1 + \liet^2
\end{align}
This gives us a pair of simultaneous equations in $\beta$ and $\gamma$.  Solving them we obtain
\begin{align}
\beta & = \frac{d}{d+1}
\\
\gamma & = \frac{1}{d\liet^2} \left(\liet^2 - \frac{1}{d+1} 
\right)
\end{align}
Substituting these expressions into Eqs.~(\ref{eq:lem9HSinner}) and~(\ref{eq:lem9IdTermsLr}) we deduce Eqs.~(\ref{eq:cor10HSInner}) and~(\ref{eq:cor10Sum}).
\end{proof}
The next lemma shows that each $\lie_r$ is a linear combination of a rank-$1$ projector and the identity:
\begin{lemma}
\label{lem:rankLprime}
Let $\lie$ be any Hermitian matrix $\in \gl(d,\mathbb{C})$ which is not a multiple of the identity. Then
\be
\rnk(\ad_\lie) \ge 2(d-1)
\ee
The lower bound is achieved if and only if $\lie$ is of the form 
\be
\lie = \eta I+ \xi P
\ee
where $P$ is a rank-$1$ projector and $\eta$, $\xi$ are any pair of  real numbers.  The eigenvalues of $\ad_\lie$ are then $\pm \xi$ (each with multiplicity $d-1$) and $0$ (with multiplicity $d^2-2 d+2$).  
\end{lemma}
\begin{proof}
Let $\lambda_1 \ge \lambda_2 \ge \dots \ge \lambda_d$ be the eigenvalues of $\lie$ arranged in decreasing order, and let $|b_1\rangle, |b_2\rangle, \dots, |b_d\rangle$ be the corresponding eigenvectors.  We may assume, without loss of generality, that the $|b_r\rangle$ are orthonormal.  We have
\be
\ad_\lie \bigl( |b_r\rangle \langle b_s | \bigr) = \bigl[\lie,|b_r\rangle \langle b_s |\bigr] = (\lambda_r - \lambda_s) |b_r\rangle \langle b_s |
\ee 
So the eigenvalues of  $\ad_\lie$ are $\lambda_r - \lambda_s$. Since $\lie$ is not a multiple of the identity we must have $\lambda_r \neq \lambda_{r+1}$ for some $r$ in the range $1\le r \le d-1$.   We then have that $\lambda_s - \lambda_t \neq 0$ if either  $s\le r < t$ or $t \le r < s$.  There are $2 r(d-r) $ such pairs $s$, $t$.  So
\be
\rnk (\ad_\lie) \ge 2r(d-r) \ge 2(d-1)
\ee
Suppose, now that the lower bound is achieved.  Then  $r(d-r) = d-1$, implying that $r=1$ or $d-1$.  Also we must have $\lambda_s = \lambda_{s+1}$ for all $s \neq r$.  So either 
\begin{align}
\lie & = \lambda_2 I + (\lambda_1-\lambda_2) |b_1\rangle \langle b_1|
\\
\intertext{or}
\lie & = \lambda_{d-1} I -(\lambda_{d-1} - \lambda_d) |b_d\rangle \langle b_d |
\end{align}
Either way $\lie$ and the spectrum of $\ad_L$ are as described.
\end{proof}
The final ingredient needed to complete the proof is 
\begin{lemma}
Let $\lie_r$,  $\str_{rst}$  and $\str_r$ be as defined in the statement of Theorem~\ref{thm:QQTEqualsSICness}.  Suppose that the $\str_r$ have the spectral decomposition
\be
\str\vpu{T}_r = \liq\vpu{T}_r - \liq^{\mathrm{T}}_r
\ee
where $\liq_r$ is a rank $d-1$ projector which is orthogonal to its own transpose.  Let $\liet$, $\sign'_r$ be as in the statement of Corollary~
\ref{corr:HSforQQTproperties}.  Then there is a fixed sign $\sign=\pm 1$ such that 
\be
\Pi_r = \sign \sign'_r \lie_r -\frac{\sign \liet -1}{d} I
\label{eq:lem12res}
\ee
is a rank-$1$ projector for all $r$.  
\end{lemma}
\begin{proof}
Define
\be
\lie'_r = \sign'_r \lie_r - \frac{\liet-1}{d} I 
\ee
Then it follows from Corollary~\ref{corr:HSforQQTproperties} that
\begin{align}
\Tr(\lie'_r) & = 1
\\
\intertext{for all $r$,}
\Tr(\lie'_r \lie'_s) & = \frac{d\delta_{rs}+1}{d+1}
\label{eq:lem12interA}
\\
\intertext{for all $r,s$, and}
\sum_{r=1}^{d^2} \lie'_r & = d I 
\label{eq:lem12interB}
\end{align}
It is also easily seen that if we define $\str'_{rst} = \sign'_r\sign'_s\sign'_t \str_{rst}$  then
\begin{align}
[\lie'_r , \lie'_s] & = \sum_{t=1}^{d^2} \str'_{rst} \lie'_t 
\\
\intertext{and} 
\str'_r &= P'_r - {P'_r }^{\mathrm{T}}
\end{align}
where $P'_r$ is a rank-$1$ projector which is orthogonal to its own transpose (see the first part of the proof of Theorem~\ref{thm:QQTEqualsSICness}).  In particular
\be
\rnk \left(\ad_{\lie'_r}\right) = 2(d-1)
\ee
and the eigenvalues of $\ad_{\lie'_r}$  all equal to $\pm 1$ or $0$.  So, taking account of the fact that $\Tr(\lie'_r)=1$, we can use Lemma~\ref{lem:rankLprime} to deduce that there is a family of rank-$1$ projectors $\Pi'_r$ and signs $\xi_r=\pm 1$ such that
\be
\lie'_r =\xi_r \Pi'_r +\frac{1-\xi_r}{d} I
\ee
If $\xi_r=+1$ (respectively $-1$) for all $r$ then Eq.~(\ref{eq:lem12res})  holds with $\Pi_r=\Pi'_r$ and $\sign= +1$ (respectively $-1$).  Also, if $d=2$ then $\lie'_r$ is  a rank-$1$ projector irrespective of the value of $\xi_r$, so Eq.~~(\ref{eq:lem12res}) holds with $\Pi_r = \lie'_r$ and $\sign=+1$. The problem therefore reduces to showing that if $d>2$ it cannot happen that $\xi_r= +1$ for some values of $r$ and $-1$ for others.  We will do this by assuming the contrary and deducing a contradiction.

Let $m$ be the number of values of $r$ for which $\xi_r = +1$.  We are assuming that $m$ is in the range $1\le m \le d^2-1$.  We may also assume, without loss of generality, that the labelling is such that $\xi_r=+1$ for the first $m$ values of $r$, and $-1$ for the rest.  So
\be
\lie'_r = \begin{cases} \Pi'_r \qquad & \text{if $r \le  m$}
\\
\frac{2}{d}I - \Pi'_r \qquad & \text{if $r > m$}
\end{cases}
\label{eq:lem12interC}
\ee 

Now define
\be
\tilde{T}_{rst} = \Tr\bigl(\lie'_r \lie'_s \lie'_t \bigr)
\ee
Eqs.~(\ref{eq:lem12interA}) and~(\ref{eq:lem12interB})
mean that the same argument which led to Eq.~(\ref{eq:PirPisProduct}) can be used to deduce
\be
\lie'_r \lie'_s = \frac{d+1}{d} \left( \sum_{t=1}^{d^2} \tilde{T}_{rst} \lie'_t \right) - K^2_{rs} I 
\ee
Since $\lie'_1$ is a projector it follows that 
\be
\lie'_1 \lie'_s = \bigl(\lie'_1\bigr)^2 \lie'_s = \frac{d+1}{d} \left(\sum_{t=1}^{d^2} \tilde{T}_{1st} \lie'_1 \lie'_t\right) - K^2_{1s} \lie'_1
\ee
By essentially the same argument which led to Eq.~(\ref{eq:TsqRel}) we can use this to infer
\be
\bigl(\tilde{T}'_1\bigr)^2 = \frac{d}{d+1} \tilde{T}\vpu{2}_1 + \frac{2 d^2}{(d+1)^2} \| \ev_1\dr \dl \ev_1 \|
\ee
where $\tilde{T}'_1$ is the matrix with matrix elements $\tilde{T}'_{1rs}$ and $\| e_1\dr$ is the vector defined by Eq.~(\ref{eq:erVecDef}).
As before $\| \ev_1\dr$ is an eigenvector of $\tilde{T}'_1$ with eigenvalue $\frac{2d}{d+1}$.  Consequently the matrix
\be
\tilde{Q}_1 = \frac{d+1}{d} \tilde{T}'_1 - 2 \| \ev_1\dr \dl \ev_1 \|
\ee
is a  projector.  But  that means $\Tr(\tilde{Q}_1)$ must be an integer.  We now use this to derive a contradiction.

It follows from  Eq.~(\ref{eq:lem12interC}) that
\be
(\lie'_r)^2 = 
\begin{cases}
\lie'_r \qquad & r\le m
\\
\frac{2(d-2)}{d^2} I - \frac{d-4}{d} \lie'_r \qquad & r > m
\end{cases}
\ee
Consequently
\be
\tilde{T}_{1rr} = \begin{cases} K^2_{1r} \qquad & r\le m 
\\ \frac{2(d-2)}{d^2} - \frac{d-4}{d} K^2_{1r} \qquad & r > m
\end{cases}
\ee
and so
\begin{align}
\Tr(\tilde{Q}_1) & = \frac{d+1}{d} \sum_{r=1}^{d^2} \tilde{T}_{1rr} - 2
\nonumber
\\
& = d+1 - \frac{4d^2 + 2m(d-2)}{d^3}
\end{align}
So if $\Tr(\tilde{Q}_1)$ is an integer $\left(4d^2 + 2n(d-2)\right)/d^3$ must also be an integer.  But the fact that $1\le m < d^2$, together with the fact that $d>2$ means
\be
\frac{4}{d} < \frac{4d^2 + 2 m(d-2)}{d^3} < 2
\ee
If $d=3$ or $4$ there are no integers in this interval, which gives us a contradiction straight away.  If, on the other hand, $d\ge 5$ there is  the  possibility
\be
\frac{4d^2 + 2 m(d-2)}{d^3}  = 1
\ee
implying 
\be
m = \frac{d^2(d-4)}{2(d-2)}
\label{eq:lem12interD}
\ee
This equation has the solution $d=6$, $m=9$ (this is in fact the only integer solution, as can be seen from an analysis of the possible prime factorizations of the numerator and denominator on the right hand side).  
To eliminate this possibility  define
\be
\lie''_r = \frac{2}{d} I - \lie'_{d^2+1 -r}
\ee
for all $r$.  It is easily verified that
\begin{align}
\Tr(\lie''_r\lie''_s) & = \frac{d\delta_{rs}+1}{d+1}
\\
\sum_{r=1}^{d^2} \lie''_r & = d I 
\intertext{and}
\lie''_r & = 
\begin{cases}
\Pi_r \qquad & r \le d^2 -m 
\\
\frac{2}{d} I - \Pi_r \qquad & r > d^2 - m
\end{cases}
\end{align}
So we can go through  the same argument as before to deduce
\be
d^2 -m = \frac{d^2(d-4)}{2(d-2)}
\label{eq:lem12interE}
\ee
Eqs.~(\ref{eq:lem12interD}) and~(\ref{eq:lem12interE}) have no joint solutions at all with $d\neq 0$, integer or otherwise.
\end{proof}
To complete the proof of Theorem~\ref{thm:QQTEqualsSICness} observe that Eqs.~(\ref{eq:cor10HSInner}) and~(\ref{eq:lem12res}) imply
\be
\Tr(\Pi_r\Pi_s) = \frac{d\delta_{rs}+1}{d+1}
\ee
So the $\Pi_r$ are a SIC-set.  Moreover
\be
L_r = \sign_r \left( \Pi_r + \alpha I\right)
\ee
where $\sign_r = \sign \sign'_r$ and $\alpha=  (\sign\liet - 1)/d$.

\section{The Algebra \texorpdfstring{${\rm sl}(d,\mathbb{C})$}{sl(d,C)}}
\label{eq:SICSandsldC}

The motivation for this paper is the hope that a Lie algebraic perspective may cast some light on the SIC-existence problem, and on the mathematics of SIC-POVMs generally.  We have focused on $\gl(d,\mathbb{C})$ as that is the case where the connection with Lie algebras seems most straightforward.  However, it may be worth mentioning that a SIC-POVM also gives rise to an interesting geometrical structure in $\sla(d,\mathbb{C})$ (the Lie algebra consisting of all trace-zero $d\times d$ complex matrices).  

Let $\Pi_r$ be a SIC-set and define
\be
\slGen_r = \sqrt{\frac{d+1}{2(d^2-1)}} \left( \Pi_r - \frac{1}{d}I\right)
\ee 
So $\slGen_r \in \sla(d,\mathbb{C})$. Let
\be
\langle A, A'\rangle = \Tr(\ad_{A}\ad_{A'}) = 2 d\Tr(A A')
\ee
be the Killing form~\cite{Lie4} on $\sla(d,\mathbb{C})$.  Then 
\be
\langle B_r, B_s\rangle =
\begin{cases}
1 \qquad & r = s
\\
-\frac{1}{d^2-1} \qquad & r\neq s
\end{cases}
\ee
So the $\slGen_r$ form a regular simplex in $\sla(d,\mathbb{C})$.  Since  $\sla(d,\mathbb{C})$ is $d^2-1$ dimensional the $\slGen_r$ are an overcomplete set.  However, the fact that 
\be
\sum_{r=1}^{d^2} \slGen_r = 0
\ee
means that for each $A\in\sla(d,\mathbb{C})$ there is a unique set of numbers $a_r$ such that
\begin{align}
A & = \sum_{r=1}^{d^2} a_r \slGen_r
\\
\intertext{and}
\sum_{r=1}^{d^2} a_r & = 0
\\
\intertext{The $a_r$ can be calculated using}
a_r & = \frac{d^2-1}{d^2} \langle A, \slGen_r\rangle 
\end{align}
Similarly, given any linear transformation $M \colon \sla(d,\mathbb{C}) \to \sla(d,\mathbb{C})$, there is a unique set of numbers $M_{rs}$ such that
\begin{align}
M \slGen_r & = \sum_{s=1}^{d^2} M_{rs} \slGen_s
\\
\intertext{and}
\sum_{s=1}^{d^2} M_{rs}& = \sum_{s=1}^{d^2} M_{sr} = 0
\\
\intertext{for all $r$.  The $M_{rs}$ can be calculated using}
M_{rs} &= \frac{d^2-1}{d^2} \langle B_s , M B_r\rangle
\end{align}
In short, the $\slGen_r$ retain many analogous properties of, and can be used in much the same way as, a basis.  It could be said that they form a simplicial basis.

\section{Further Identities}
\label{sec:AdjRepFurther}

In the preceding pages we have seen that there are five different families of matrices naturally associated with a SIC-POVM:  namely, the projectors $Q^{\vpu{T}}_r$  together with the matrices
\begin{align}
J_r & = Q^{\vpu{T}}_r -Q^{\mathrm{T}}_r
\\
\bar{R}_r & = Q^{\vpu{T}}_r +Q^{\mathrm{T}}_r
\\
R_r & =  Q^{\vpu{T}}_r +Q^{\mathrm{T}}_r + 4 \|e_r \dr \dl e_r \|
\\
T_r & = \frac{d}{d+1} Q_r + \frac{2d}{d+1} \|e_r \dr \dl e_r \|
\end{align}
(see Section~\ref{sec:QQTProp1}).  As we noted previously, it is possible to define everything in terms of the adjoint representation matrices $J_r$ and the rank-$1$ projectors $\|e_r\dr \dl e_r \|$:
\begin{align}
Q_r & = \frac{1}{2} J_r (J_r + I) 
\\
\bar{R}_r & = J^2_r
\\
R_r & = J^2_r + 4 \| e_r \dr \dl e_r \|
\\
T_r & = \frac{d}{2(d+1)} J_r (J_r+I) + \frac{2d}{d+1} \|e_r \dr \dl e_r \|
\end{align}
In that sense the structure constants of the Lie algebra, supplemented with the vectors $\|e_r \dr$, determine everything else.

In the next section we will show that there are some interesting geometrical relationships between the hyperplanes onto which $Q_r$, $Q^{\mathrm{T}}_r$ and $\bar{R}_r$ project.  In this section, as a preliminary to that investigation, we prove a number of identities satisfied by the $Q$, $J$ and $\bar{R}$ matries.
We start  by computing their Hilbert-Schmidt inner products:
\begin{theorem}
For all $r,s$
\begin{align}
\Tr\bigl( Q_r Q_s \bigr) 
& = 
\frac{d^3 \delta_{rs} +d^2-d-1}{(d+1)^2}
\\
\Tr\bigl(Q^{\vpu{T}}_r Q^{\mathrm{T}}_s\bigr)
& =\frac{d^2(1-\delta_{rs})}{(d+1)^2}
\\
\Tr\bigl(J_r J_s \bigr)
& =
\frac{2(d^2\delta_{rs} -1)}{d+1}
\\
\Tr\bigl(\bar{R}_r \bar{R}_s\bigr)
& =
\frac{2(d-1)(d^2\delta_{rs} + 2 d +1)}{(d+1)^2}
\\
\Tr\bigl(J_r \bar{R}_s \bigr) & = 0 
\end{align}
\end{theorem}

\begin{proof}
Let us first calculate some auxiliary quantities.  It follows from the definition of $T_{r}$, and the fact that the matrix $P = \frac{1}{d} G$ defined by Eq.~(\ref{eq:SICGramProjectorDef}) is a rank $d$ projector, that
\begin{align}
\Tr(T_r T_s) 
& = 
\sum_{u,v=1}^{d^2} T_{ruv}T_{svu}
\nonumber
\\
& = 
\sum_{u,v=1}^{d^2} K^2_{uv} G_{ru}G_{us} G_{sv}G_{vr}
\nonumber
\\
& =
\frac{d}{d+1}\sum_{u=1}^{d^2} K^2_{ru}K^2_{su} + \frac{d^4}{d+1} \sum_{u,v=1}^{d^2} P_{ru}P_{us} P_{sv}P_{vr}
\nonumber
\\
& = \frac{d^2(d\delta_{rs} + d + 2)}{(d+1)^3}  + \frac{d^4}{d+1}  \bigl|P_{rs}\bigr|^2
\nonumber
\\
& = \frac{d^2(d\delta_{rs} + d + 2)}{(d+1)^3}+\frac{d^2}{d+1}  K^2_{rs}
\nonumber
\\
& = \frac{d^2\bigl(d(d+2) \delta_{rs} + 2 d + 3  \bigr)}{(d+1)^3}
\end{align}

Also 
\begin{align}
\Tr\bigl(T^{\vpu{T}}_r T^{\mathrm{T}}_s\bigr)
& = 
\sum_{u,v=1}^{d^2} T_{ruv}T_{suv}
\nonumber
\\
& =
\sum_{u=1}^{d^2}G_{ru}G_{su}\left(\sum_{v=1}^{d^2} G_{uv}G_{uv} G_{vr}G_{vs}\right)
\nonumber
\\
& =\frac{2d}{d+1} \sum_{u=1}^{d^2} G_{ru}G_{su} G_{ur}G_{us}
\nonumber
\\
& = \frac{2d^2}{(d+1)^2} \left( 1+ K^2_{rs} \right)
\nonumber
\\
& = \frac{2d^2(d\delta_{rs} + d+2)}{(d+1)^3}
\end{align}
where we made two applications of Eq.~(\ref{eq:2designProp}) (\emph{i.e.}\ the fact that every SIC-POVM is a $2$-design). Finally, it is a straightforward consequence of the definitions of $T^{\vpu{T}}_{r}$, $T^{\mathrm{T}}_r$ and $\|e_r\dr$ that 
\begin{align}
\dl e_r \| T^{\vpu{T}}_s \| e_r \dr & = \dl e_r \| T^{\mathrm{T}}_s \| e_r \dr 
\nonumber
\\
& =
\frac{d+1}{2d}\sum_{u,v=1}^{d^2} T^{\vpu{2}}_{suv} K^2_{ru} K^2_{rv}
\nonumber
\\
& = \frac{1}{2d(d+1)} \left(d^2 T_{srr} + d \sum_{v=1}^{d^2} T_{srv} + d\sum_{u=1}^{d^2} T_{sur} + \sum_{u,v=1}^{d^2} T_{suv}
\right)
\nonumber
\\
& = \frac{d}{2(d+1)} \left(3  K^2_{rs}  + 1
\right)
\nonumber
\\
& = \frac{d(3 d\delta_{rs} + d+ 4)}{2(d+1)^2}
\\
\intertext{and}
\dl e_r \| e_s \dr 
& = 
\frac{d+1}{2d} \sum_{u=1}^{d^2} K^2_{ru} K^2_{su}
\nonumber
\\
& = \frac{d\delta_{rs}+d+2}{2(d+1)}
\end{align}
Using these results in the expressions
\begin{align}
\Tr\bigl(Q_r Q_s\bigr) 
&=
\Tr\Biggl(\left(\frac{d+1}{d} T_r - 2 \| e_r \dr \dl e_r \| 
\right)\left(\frac{d+1}{d} T_s - 2 \| e_s \dr \dl e_s \| 
\right)
\Biggr)
\\
\intertext{and}
\Tr\bigl(Q^{\vpu{T}}_r Q^{\mathrm{T}}_s\bigr) 
&=
\Tr\Biggl(\left(\frac{d+1}{d} T^{\vpu{T}}_r - 2 \| e_r \dr \dl e_r \| 
\right)\left(\frac{d+1}{d} T^{\mathrm{T}}_s - 2 \| e_s \dr \dl e_s \| 
\right)
\Biggr)
\end{align}
the first two statements follow.  The remaining statements are immediate consequences of these and the fact that
\begin{align}
J^{\vpu{T}}_r &= Q^{\vpu{T}}_r - Q^{\mathrm{T}}_r
\\
\bar{R}^{\vpu{T}}_r & = Q^{\vpu{T}}_r + Q^{\mathrm{T}}_r
\end{align}

\end{proof}
Now define 
\be
\| \vv \dr = \frac{1}{d} \sum_{r=1}^{d^2} \| r \dr
\label{eq:vZeroDef}
\ee
where $\|r\dr$ is the basis defined in Eq.~(\ref{eq:hd2StandardBasis}).  The following result shows (among other things) that the subspaces onto which the $Q^{\vpu{T}}_r$ (respectively $Q^{\mathrm{T}}_r$, $R^{\vpu{T}}_r$) project  span the orthogonal complement of $\|\vv\dr$.  
\begin{theorem}
\label{thm:vVecIdentities}
For all $r$
\be
Q^{\vpu{T}}_r \| \vv\dr = Q^{\mathrm{T}}_r \| \vv\dr = J^{\vpu{T}}_r \| \vv\dr = R^{\vpu{T}}_r \| \vv\dr = 0
\label{eq:vVecAnnihalation}
\ee
Moreover
\begin{align}
\sum_{r=1}^{d^2} Q^{\vpu{T}}_r = \sum_{r=1}^{d^2} Q^{\mathrm{T}}_r  & = \frac{d^2}{d+1} \bigl( I - \| \vv\dr \dl \vv \| \bigr)
\label{eq:Qsumrel}
\\
\sum_{r=1}^{d^2} J_r & = 0
\\
\sum_{r=1}^{d^2} \bar{R}_r & = \frac{2d^2}{d+1} \bigl( I - \| \vv\dr \dl \vv \| \bigr)
\label{eq:Rbsumrel}
\end{align}
\end{theorem}
\begin{proof}
Some of this is a straightforward consequence of the fact that $J_r$ is the adjoint representative of $\Pi_r$.  Since
\be
\sum_{s=1}^{d^2} \Pi_s = d I
\ee
we must have
\be
\sum_{s,t=1}^{d^2}  J_{rst} \Pi_t  = \sum_{s=1}^{d^2} \ad_{\Pi_r} \Pi_s = 0
\ee
In view of the antisymmetry of the $J_{rst}$ it follows that
\begin{align}
\sum_{r=1}^{d^2} J_r & = 0
\\
\intertext{and}
J_r \| \vv\dr & = 0
\end{align}
Using the relations 
\begin{align}
Q^{\vpu{T}}_r &= \frac{1}{2} J^{\vpu{T}}_r(J^{\vpu{T}}_r+I)
\\
Q^{\mathrm{T}}_r & = \frac{1}{2} J^{\vpu{T}}_r(J^{\vpu{T}}_r-I)
\\
\bar{R}^{\vpu{T}}_r & = J^{2}_r
\end{align}
we deduce 
\be
Q^{\vpu{T}}_r \|\vv\dr = Q^{\mathrm{T}}_r \| \vv\dr = \bar{R}^{\vpu{T}}_r \| \vv\dr = 0
\ee
It remains to prove Eqs.~(\ref{eq:Qsumrel}) and~(\ref{eq:Rbsumrel}).  It follows from Eq.~(\ref{eq:Qdef}) that
\begin{align}
\sum_{r=1}^{d^2} Q_{rst} &= \frac{d+1}{d} \sum_{r=1}^{d^2} T_{rst} - 2 \sum_{r=1}^{d^2} \dl s \| e_r \dr \dl e_r \| t\dr 
\nonumber
\\
& = (d+1) K^2_{st} - \frac{d+1}{d}\sum_{r=1}^{d^2} K^2_{rs}K^2_{rt}
\nonumber
\\
& = \frac{d^2\delta_{st}-1}{d+1}
\end{align}
from which it follows 
\be
\sum_{r=1}^{d^2} Q^{\vpu{T}}_r = \sum_{r=1}^{d^2} Q^{\mathrm{T}}_r=\frac{d^2}{d+1} \bigl( I - \| \vv \dr \dl \vv \| \bigr)
\ee
Eq.~(\ref{eq:Rbsumrel})  follows from this and the fact that $R^{\vpu{T}}_r = Q^{\vpu{T}}_r + Q^{\mathrm{T}}_r$.

\end{proof}

\section{Geometrical Considerations}
\label{sec:geometry}

In this section we show that  there are some interesting geometrical relationships between the subspaces onto which the operators $Q^{\vpu{T}}_r$, $Q^{\mathrm{T}}_r$ and $\bar{R}_r$ project.  The original motivation for this work was an observation concerning the subspaces onto which the $\bar{R}_r$ project.  $\bar{R}_r$ is a real matrix, and so it defines a $2(d-2)$ subspace in $\mathbb{R}^{d^2}$, which we will denote $\mathcal{R}_r$.  We noticed that for each pair of distinct indices $r$ and $s$ the intersection $\mathcal{R}_r\cap \mathcal{R}_s$ is a $1$-dimensional line.  This led us to speculate that a set of hyperplanes parallel to the $\mathcal{R}_r$ might be the edges of an interesting polytope.  We continue to think that this could  be the case.  Unfortunately we have not been able to prove it.  However, it appears to us that the results we obtained while trying to prove it have an interest which is independent of the truth of the motivating speculation.

We will begin with some terminology.  Let $P$ be any projector (on either $\mathbb{R}^N$ or $\mathbb{C}^N$),  let $\mathcal{P}$ be the subspace onto which $P$ projects, and let $|\psi\rangle$ be any non-zero vector.  Then we define the angle between $|\psi\rangle$ and $\mathcal{P}$ in the usual way, to be
\be
\theta = \cos^{-1} \left( \frac{\bigl\| P |\psi\rangle \bigr\| }{\bigl\| |\psi \rangle \bigr\|}\right)
\ee
(so $\theta$ is the smallest angle between $|\psi \rangle$ and any of the vectors in $\mathcal{P}$).  

Suppose, now, that $P'$ is another projector, and let $\mathcal{P}'$ be the subspace onto which $P'$ projects.  We will say that  $\mathcal{P}'$ is uniformly inclined to $\mathcal{P}$ if every vector in $\mathcal{P}'$ makes the same angle $\theta$ with $\mathcal{P}$.  If $\theta=0$ this means that $\mathcal{P}' \subseteq \mathcal{P}$, while if $\theta=\frac{\pi}{2}$ it means $\mathcal{P}'\perp \mathcal{P}$.  Suppose, on the other hand, that $0< \theta < \frac{\pi}{2}$.  Let $|u'_1\rangle, \dots , |u'_n\rangle$ be any  orthonormal basis for $\mathcal{P}'$, and define $|u_r\rangle = \sec \theta P |u'_r\rangle$.  Then $\langle u_r | u_r \rangle = 1$ for all $r$.  Moreover, if $P, P'$ are complex projectors, 
\be
\langle u'_r + e^{i\phi} u'_s| P |u'_r + e^{i\phi} u'_s\rangle  
 = 
2\cos^2 \theta\Bigl(1 +  \rel\left(  e^{i\phi} \langle u_r | u_s \rangle
\right)\Bigr)
\ee
for all $\phi$ and $r\neq s$.
On the other hand it follows from the assumption that $\mathcal{P}'$ is uniformly inclined to $\mathcal{P}$ that
\be
\langle u'_r + e^{i\phi} u'_s| P |u'_r + e^{i\phi} u'_s\rangle  = 2 \cos^2 \theta 
\ee
for all $\phi$ and $r\neq s$.  Consequently
\be
\langle u_r | u_s \rangle = \delta_{rs}
\ee
for all $r,s$.   It is easily seen that the same is true if $P, P'$ are real projectors.  

Suppose we now make the further assumption that $\dim(\mathcal{P}') = \dim (\mathcal{P}) = n$.  Then $|u_1\rangle, \dots , |u_n\rangle$ is an orthonormal basis for $\mathcal{P}$, and we can write
\begin{align}
P& = \sum_{r=1}^{n} |u_r\rangle \langle u_r |
\label{eq:PtermsBasis}
\\
P'& = \sum_{r=1}^{n} |u'_r\rangle \langle u'_r |
\label{eq:PptermsBasis}
\end{align}
Observe that
\be
\langle u'_r | u_s \rangle = \langle u'_r | P |u_s\rangle = \cos \theta \langle u_r | u_s \rangle  = \cos \theta \delta_{rs}
\ee
for all $r,s$. 
Consequently
\be
P'|u_r\rangle  = \cos \theta | u_r\rangle
\ee
for all $r$.  It follows that
\be
\bigl\| P' |\psi\rangle \bigr\|
=
\left\|\sum_{r=1}^{n} \cos\theta \langle u_r |\psi\rangle |u'_r\rangle\right\|
=
\cos \theta \bigl\| |\psi\rangle \bigr\|
\ee
 for all $|\psi\rangle\in \mathcal{P}$.  So $\mathcal{P}$ is uniformly inclined to $\mathcal{P}'$ at the same angle $\theta$.

It follows from Eqs.~(\ref{eq:PtermsBasis}) and~(\ref{eq:PptermsBasis}) that
\begin{align}
P P' P & = \cos^2\theta P
\label{eq:PPpPcond}
\\
P' P P' & = \cos^2 \theta P'
\label{eq:PpPPpcond}
\end{align}

Eq.~(\ref{eq:PPpPcond}), or equivalently Eq.~(\ref{eq:PpPPpcond}), is not only necessary but also sufficient for the subspaces to be uniformly inclined.  In fact,
let $\mathcal{P}$, $\mathcal{P}'$ be any two subspaces which have the same dimension $n$, but which are not assumed at the outset to be uniformly inclined, and let $P$, $P'$ be the corresponding projectors.  Suppose
\be
P P' P = \cos^2 \theta P
\ee
for some $\theta$ in the range $0 \le \theta \le \frac{\pi}{2}$.  It is immediate that $\mathcal{P} = \mathcal{P}'$ if $\theta=0$, and $\mathcal{P}\perp \mathcal{P}'$ if $\theta= \frac{\pi}{2}$.  Either way, the subspaces are uniformly inclined.  Suppose, on the other hand, that $0 < \theta < \frac{\pi}{2}$.  Let $|u'_1\rangle, \dots , |u'_n\rangle$ be any orthonormal basis for $\mathcal{P}'$, and define $|u_r\rangle = \sec\theta P |u'_r\rangle$.
Eq.~(\ref{eq:PtermsBasis}) then implies
\be
P = \sec^2\theta \sum_{r=1}^{n} P |u'_r\rangle \langle u'_r | P = \sum_{r=1}^{n} |u_r\rangle \langle u_r |
\ee
Given any $|\psi\rangle \in \mathcal{P}$ we have
\be
|\psi \rangle = P | \psi \rangle = \sum_{r=1}^n \langle u_r | \psi \rangle |u_r\rangle
\ee
Since $\dim(\mathcal{P}) = n$ it follows that the $|u_r\rangle$ are linearly independent.  In particular
\be
|u_r\rangle = P |u_r\rangle = \sum_{s=1}^{n} \langle u_s | u_r \rangle |u_s\rangle
\ee
Since the $|u_r\rangle$ are linearly independent this means
\be
\langle u_s | u_r \rangle = \delta_{rs}
\ee
So the $|u_r\rangle$ are an orthonormal basis for $\mathcal{P}$.  
It follows, that if $|\psi'\rangle$ is any vector in $\mathcal{P}'$, then
\be
\bigl\| P |\psi'\rangle \bigr\| = \left\| \sum_{r=1}^{n} \langle u'_r | \psi'\rangle P |u'_r\rangle
\right\|
=\cos \theta \left\| \sum_{r=1}^{n} \langle u'_r | \psi'\rangle |u_r\rangle 
\right\|
= \cos \theta \bigl\| |\psi'\rangle \bigr\|
\ee
implying that $\mathcal{P}'$ is uniformly inclined to $\mathcal{P}$ at angle $\theta$.  

It will be convenient to summarise all this in the form of a lemma:
\begin{lemma}
\label{lem:UniformIncline}
Let $\mathcal{P}$, $\mathcal{P}'$ be any two subspaces, real or complex, having the same dimension $n$.  Let $P$, $P'$ be the corresponding projectors.  Then the following statements are equivalent:
\begin{enumerate}
\item[(a)] $\mathcal{P}'$ is uniformly inclined to $\mathcal{P}$ at angle $\theta$.
\item[(b)]  $\mathcal{P}$ is uniformly inclined to $\mathcal{P}'$ at angle $\theta$.
\item[(c)] 
\be
P P' P = \cos^2 \theta P 
\ee
\item[(d)]
\be
P' P P' = \cos^2 \theta P'
\ee
\end{enumerate}
Suppose these conditions are satisfied for some $\theta$ in the range $0< \theta < \frac{\pi}{2}$, and let $|u_1\rangle, \dots |u_n\rangle$ be any orthonormal basis for $\mathcal{P}$.  Then there exists an orthonormal basis $|u'_1\rangle, \dots, |u'_n\rangle$ for $\mathcal{P}'$ such that
\begin{align}
P'| u_r\rangle &= \cos \theta |u'_r\rangle
\\
P | u'_r\rangle & = \cos \theta | u_r \rangle 
\end{align}
\end{lemma}

We are now in a position to state the main results of this section.  Let $\mathcal{Q}_r$ (respectively $\bar{\mathcal{Q}}_r$) be the subspace onto which $Q_r$ (respectively $Q^{\mathrm{T}}_r$) projects. We then have
\begin{theorem}
\label{thm:Qgeometry}
For each pair of distinct indices $r$, $s$ the subspaces $\mathcal{Q}_r$, $\bar{\mathcal{Q}}_r$ have the orthogonal decomposition
\begin{align}
\mathcal{Q}^{\vpu{0}}_r & = \mathcal{Q}^{0}_{rs} \oplus  \mathcal{Q}^{\vpu{0}}_{rs} 
\\
\bar{\mathcal{Q}}^{\vpu{0}}_r & = \bar{\mathcal{Q}}^{0}_{rs} \oplus  \bar{\mathcal{Q}}^{\vpu{0}}_{rs} 
\end{align}
where 
\begin{align}
\mathcal{Q}^{0}_{rs} &\perp \mathcal{Q}^{\vpu{0}}_{rs} & \dim(\mathcal{Q}^{0}_{rs}) &=1 & \dim(\mathcal{Q}^{\vpu{0}}_{rs}) & = d-2
\nonumber
\\
\bar{\mathcal{Q}}^{0}_{rs}& \perp \bar{\mathcal{Q}}^{\vpu{0}}_{rs} & \dim(\bar{\mathcal{Q}}^{0}_{rs}) &= 1 & \dim(\bar{\mathcal{Q}}^{\vpu{0}}_{rs}) &= d-2
\nonumber
\end{align}
We have
\begin{enumerate}
\item[(a)] Relation of the subspaces $\mathcal{Q}_r$ and $\mathcal{Q}_s$:
\begin{enumerate}
\item[(1)] $\mathcal{Q}^0_{rs} \perp \mathcal{Q}^{\vpu{0}}_{sr}$ and $\mathcal{Q}^{\vpu{0}}_{rs} \perp \mathcal{Q}^0_{sr}$.
\item[(2)] $\mathcal{Q}^0_{rs}$ and $\mathcal{Q}^0_{sr}$ are inclined at angle $\cos^{-1} \bigl(\frac{1}{d+1}\bigr)$.
\item[(3)]  $\mathcal{Q}_{rs}$ and $\mathcal{Q}_{sr}$ are uniformly inclined at angle $\cos^{-1} \Bigl(\frac{1}{\sqrt{d+1}}\Bigr)$.
\end{enumerate}
\item[(b)] Relation of the subspaces $\bar{\mathcal{Q}}_r$ and $\bar{\mathcal{Q}}_s$:
\begin{enumerate}
\item[(1)] $\bar{\mathcal{Q}}^0_{rs} \perp \bar{\mathcal{Q}}^{\vpu{0}}_{sr}$ and $\bar{\mathcal{Q}}^{\vpu{0}}_{rs} \perp \bar{\mathcal{Q}}^0_{sr}$.
\item[(2)] $\bar{\mathcal{Q}}^0_{rs}$ and $\bar{\mathcal{Q}}^0_{sr}$ are inclined at angle $\cos^{-1} \bigl(\frac{1}{d+1}\bigr)$.
\item[(3)]  $\bar{\mathcal{Q}}_{rs}$ and $\bar{\mathcal{Q}}_{sr}$ are uniformly inclined at angle $\cos^{-1} \Bigl(\frac{1}{\sqrt{d+1}}\Bigr)$.
\end{enumerate}
\item[(c)] Relation of the subspaces $\mathcal{Q}_r$ and $\bar{\mathcal{Q}}_s$:
\begin{enumerate}
\item[(1)] $\mathcal{Q}^0_{rs} \perp \bar{\mathcal{Q}}^{\vpu{0}}_{sr}$,  $\mathcal{Q}^{\vpu{0}}_{rs} \perp \bar{\mathcal{Q}}^0_{sr}$ and $\mathcal{Q}^{\vpu{0}}_{rs} \perp \bar{\mathcal{Q}}^{\vpu{0}}_{sr}$.
\item[(2)] $\mathcal{Q}^0_{rs}$ and $\bar{\mathcal{Q}}^0_{sr}$ are inclined at angle $\cos^{-1} \bigl(\frac{d}{d+1}\bigr)$.
\end{enumerate}
\end{enumerate}
\end{theorem}

The relations between these subspaces are, perhaps, easier to assimilate if presented pictorially.  In the following  diagrams the line joining each pair of subspaces is labelled with the cosine of the angle between them.  In particular  a $0$ on the line joining two subspaces indicates that they are orthogonal.
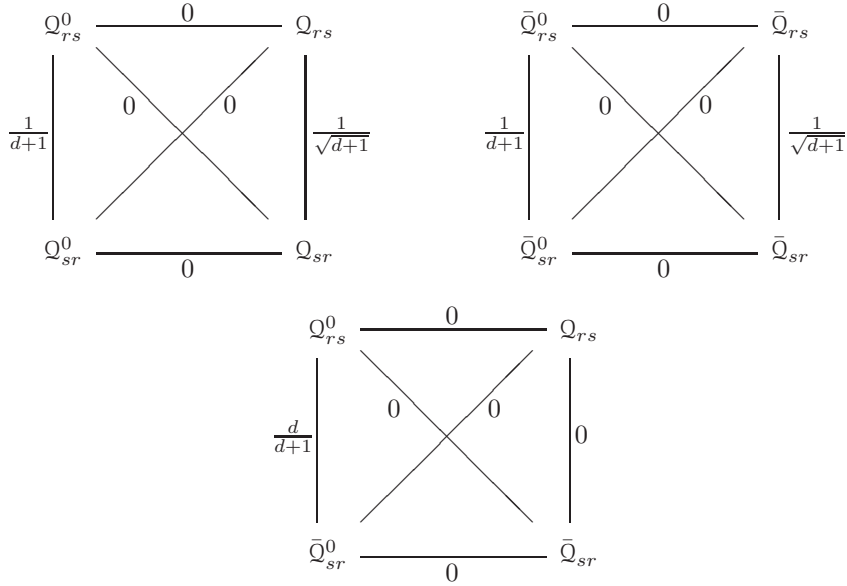
\begin{center}
\begin{picture}(355,110)
\put(20,95){$\mathcal{Q}^0_{rs}$}
\put(115,95){$\mathcal{Q}_{rs}$}
\put(20,10){$\mathcal{Q}^0_{sr}$}
\put(115,10){$\mathcal{Q}_{sr}$}
\put(40,98){\line(1,0){70}}
\put(72,100){$0$}
\put(40,12){\line(1,0){70}}
\put(72,3){$0$}
\put(23.5,25){\line(0,1){62}}
\put(6,55){$\frac{1}{d+1}$}
\put(119,25){\line(0,1){62}}
\put(121,55){$\frac{1}{\sqrt{d+1}}$}
\put(40,25){\line(1,1){65}}
\put(88,65){0}
\put(40,90){\line(1,-1){65}}
\put(50,65){0}
\put(200,95){$\bar{\mathcal{Q}}^0_{rs}$}
\put(295,95){$\bar{\mathcal{Q}}_{rs}$}
\put(200,10){$\bar{\mathcal{Q}}^0_{sr}$}
\put(295,10){$\bar{\mathcal{Q}}_{sr}$}
\put(220,98){\line(1,0){70}}
\put(252,100){$0$}
\put(220,12){\line(1,0){70}}
\put(252,3){$0$}
\put(203.5,25){\line(0,1){62}}
\put(186,55){$\frac{1}{d+1}$}
\put(298,25){\line(0,1){62}}
\put(301,55){$\frac{1}{\sqrt{d+1}}$}
\put(220,25){\line(1,1){65}}
\put(268,65){0}
\put(220,90){\line(1,-1){65}}
\put(230,65){0}
\end{picture}
\end{center}

\begin{center}
\begin{picture}(155,110)
\put(20,95){$\mathcal{Q}^0_{rs}$}
\put(115,95){$\mathcal{Q}_{rs}$}
\put(20,10){$\bar{\mathcal{Q}}^0_{sr}$}
\put(115,10){$\bar{\mathcal{Q}}_{sr}$}
\put(40,98){\line(1,0){70}}
\put(72,100){$0$}
\put(40,12){\line(1,0){70}}
\put(72,3){$0$}
\put(23.5,25){\line(0,1){62}}
\put(6,55){$\frac{d}{d+1}$}
\put(119,25){\line(0,1){62}}
\put(121,55){$0$}
\put(40,25){\line(1,1){65}}
\put(88,65){0}
\put(40,90){\line(1,-1){65}}
\put(50,65){0}
\end{picture}
\end{center}

We will prove this theorem below.  Before doing so, however, let us state the other main result of this section.
Let  $\mathcal{R}_r$ be the subspace onto which the $\bar{R}_{r}$ project.  Since $\bar{R}_{r}$ is a real matrix we regard $\mathcal{R}_r$ as a subspace of $\mathbb{R}^{d^2}$.  We have
\begin{theorem}
\label{thm:Rgeometry}
For each pair of distinct indices $r,s$ the subspace $\mathcal{R}_r$ has the decomposition
\be
\mathcal{R}_r = \mathcal{R}^0_{rs} \oplus \mathcal{R}^1_{rs} \oplus \mathcal{R}^{\vpu{0}}_{rs}
\ee
where $\mathcal{R}^0_{rs}$, $\mathcal{R}^1_{rs}$, $\mathcal{R}_{rs}$ are pairwise orthogonal and
\begin{align}
\dim(\mathcal{R}^0_{rs}) & = 1 & \dim(\mathcal{R}^1_{rs}) & = 1 & \dim(\mathcal{R}^{\vpu{0}}_{rs}) & = 2d -4
\end{align}
We have
\begin{enumerate}
\item $\mathcal{R}^0_{rs} = \mathcal{R}^0_{sr}$.
\item
$\mathcal{R}^1_{rs} \perp  \mathcal{R}^{\vpu{1}}_{sr}$ and $\mathcal{R}^{\vpu{1}}_{rs} \perp \mathcal{R}^1_{sr}$.
\item $\mathcal{R}^1_{rs}$ and $\mathcal{R}^1_{sr}$ are inclined at angle $\cos^{-1} \bigl(\frac{d-1}{d+1}\bigr)$.
\item $\mathcal{R}^{\vpu{1}}_{rs}$ and $\mathcal{R}^{\vpu{1}}_{sr}$ are uniformly inclined at angle $\cos^{-1} \Bigl(\sqrt{\frac{1}{d+1}}\Bigr)$
\end{enumerate}
In particular, the subspaces $\bar{\mathcal{R}}_r$ and $\bar{\mathcal{R}}_s$ intersect in a line.
\end{theorem}

In diagrammatic form the relations between these subspaces are
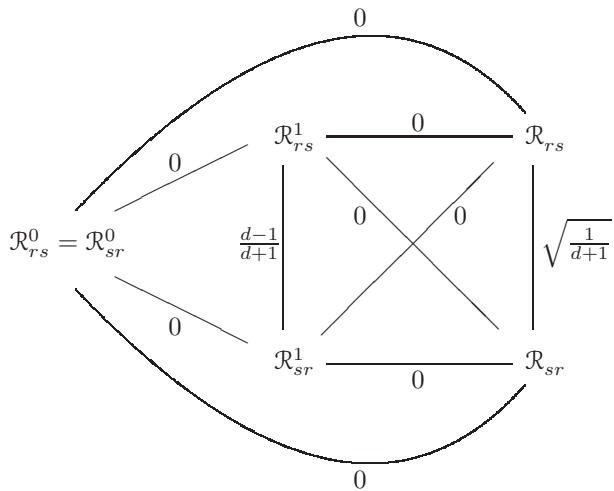
\begin{center}
\begin{picture}(355,230)
\put(20,115){$\mathcal{R}^0_{rs}=\mathcal{R}^0_{sr}$}
\put(120,155){$\mathcal{R}^1_{rs}$}
\put(215,155){$\mathcal{R}^{\vpu{1}}_{rs}$}
\put(120,70){$\mathcal{R}^1_{sr}$}
\put(215,70){$\mathcal{R}^{\vpu{1}}_{sr}$}
\put(140,158){\line(1,0){70}}
\put(172,160){$0$}
\put(140,72){\line(1,0){70}}
\put(172,63){$0$}
\put(123.5,85){\line(0,1){62}}
\put(106,115){$\frac{d-1}{d+1}$}
\put(218,85){\line(0,1){62}}
\put(221,115){$\sqrt{\frac{1}{d+1}}$}
\put(138,83){\line(1,1){65}}
\put(188,125){$0$}
\put(140,150){\line(1,-1){65}}
\put(150,125){$0$}
\put(60,130){\line(2,1){50}}
\put(80,145){$0$}
\put(60,105){\line(2,-1){50}}
\put(80,83){$0$}
\qbezier(45,130)(150,240)(215,167)
\put(150,200){$0$}
\qbezier(45,100)(150,-10)(215,64)
\put(150,25){$0$}
\end{picture}
\end{center}
where, as before, each line is labelled with the cosine of the angle between the two subspaces it connects.

\subsection*{Proof of Theorem~\ref{thm:Qgeometry}}
Let $\| 1\dr , \dots , \| d^2\dr$ be the standard basis for $\mathcal{H}_{d^2}$, as defined by Eq.~(\ref{eq:hd2StandardBasis}).  For each pair of distinct indices $r,s$ define
\begin{align}
\| f^{\vpu{*}}_{rs} \dr &= i \sqrt{d+1} Q^{\vpu{T}}_r \|s\dr
\\
\| f^{*}_{rs} \dr & = - i \sqrt{d+1} Q^{\mathrm{T}}_r\| s\dr
\end{align}
The significance of these vectors is that $\| f^{\vpu{*}}_{rs} \dr \dl f^{\vpu{*}}_{rs} \|$ (respectively $\|f^{*}_{rs}\dr \dl f^{*}_{rs} \|$) will turn out to be the projector onto the $1$-dimensional subspace $\mathcal{Q}^0_{rs}$ (respectively $\bar{\mathcal{Q}}^0_{rs}$). 

Note  that the fact that $Q_r$ is Hermitian means
\be
Q^{\mathrm{T}}_r = Q^{*}_r
\ee
(where $Q^{*}_r$ is the matrix whose elements are the complex conjugates of the corresponding elements of $Q^{\vpu{*}}_r$).  
Consequently
\be
\dl t \| f^{*}_{rs} \dr = \Bigl( \dl t \| f^{\vpu{*}}_{rs}\dr\Bigr)^{*}
\label{eq:fstarIsConjOff}
\ee 
for all $r,s,t$. 
 
It is easily seen that $\| f^{\vpu{*}}_{rs} \dr$, $\| f^{*}_{rs} \dr$ are normalized.  In fact, it follows from Eqs.~(\ref{eq:erVecDef}) and~(\ref{eq:Qdef}) that
\begin{align}
\dl f_{rs}\| f_{rs} \dr & = 
(d+1) \dl s \| Q_r \| s \dr 
\nonumber
\\
& = \frac{(d+1)^2}{d} T_{rss} - 2 (d+1)\dl s \| e_r \dr \dl e_r \| s\dr 
\nonumber
\\
& = \frac{(d+1)^2}{d} \bigl(K^2_{rs} - K^4_{rs}\bigr)
\nonumber
\\
& = 1
\label{eq:frsNormDeriv}
\end{align}
for all $r\neq s$.  In view of Eq.~(\ref{eq:fstarIsConjOff}) we then have
\be
\dl f^{*}_{rs} \| f^{*}_{rs} \dr 
=
\Bigl(\dl f^{\vpu{*}}_{rs} \| f^{\vpu{*}}_{rs} \dr
\Bigr)^{*}
= 1
\label{eq:fStarrsNormDeriv}
\ee
for all $r\neq s$.
The fact that $Q_r Q^{\mathrm{T}}_r=0$ means we also have
\be
\dl f^{\vpu{*}}_{rs} \| f^{*}_{rs} \dr = 0 
\ee
for all $r\neq s$.

Note that, although we required that $r\neq s$ in the definitions of $\| f^{\vpu{*}}_{rs}\dr$, $\| f^{*}_{rs}\dr$, the definitions continue to make sense when $r=s$.  However, the vectors are then zero (as can be seen by setting $r=s$ in Eq.~(\ref{eq:Qexplicit})).

The vectors $\|f^{\vpu{*}}_{rs}\dr$, $\| f^{*}_{rs}\dr$ satisfy a number of identities, which it will be convenient to collect in a lemma:
\begin{lemma}
\label{lem:fVecIds}
For all $r\neq s$
\begin{align}
\| f_{rs}\dr & = -\| f^{*}_{sr}\dr + i \sqrt{\frac{2}{d}} \Bigl( \| e_s\dr - \|e_r\dr
\Bigr)
\label{eq:frsTermsfStsr}
\\
\| f^{*}_{rs} \dr & = - \|f_{sr}\dr - i \sqrt{\frac{2}{d}} \Bigl(\| e_s\dr - \| e_r \dr
\Bigr)
\label{eq:fStrsTermsfsr}
\end{align}
(where $\|e_r\dr$ is the vector defined by Eq.~(\ref{eq:erVecDef}))
\begin{align}
Q^{\vpu{T}}_r \| f^{\vpu{*}}_{rs} \dr &= \| f^{\vpu{*}}_{rs}\dr 
& 
Q^{\mathrm{T}}_r \| f^{*}_{rs} \dr &= \| f^{*}_{rs} \dr
\label{eq:QfrsIdA}
\\
Q^{\mathrm{T}}_r \| f^{\vpu{*}}_{rs} \dr & = 0 
& 
Q^{\vpu{T}}_r \| f^{*}_{rs} \dr & = 0
\label{eq:QfrsIdB}
\\
Q^{\vpu{T}}_s \| f^{\vpu{*}}_{rs} \dr & = - \frac{1}{d+1} \| f^{\vpu{*}}_{sr}\dr
&
Q^{\mathrm{T}}_s \| f^{*}_{rs} \dr & = - \frac{1}{d+1} \| f^{*}_{sr} \dr
\label{eq:QfrsIdC}
\\
Q^{\mathrm{T}}_s \| f^{\vpu{*}}_{rs} \dr & = - \frac{d}{d+1} \| f^{*}_{sr}\dr
&
Q^{\vpu{T}}_s \| f^{*}_{rs}\dr & = - \frac{d}{d+1} \| f^{\vpu{*}}_{sr} \dr
\label{eq:QfrsIdD}
\end{align}
\begin{align}
\dl f^{\vpu{*}}_{rs} \| f^{\vpu{*}}_{sr} \dr 
& =
\dl f^{*}_{rs} \| f^{*}_{sr} \dr
=
-\frac{1}{d+1}
\\
\dl f^{\vpu{*}}_{rs} \| f^{*}_{sr} \dr
&=
\dl f^{*}_{rs} \| f^{\vpu{*}}_{sr} \dr
=
-\frac{d}{d+1}
\end{align}
\end{lemma}
\begin{proof}
It follows from Eqs.~(\ref{eq:erVecDef}) and~(\ref{eq:Qdef}) that
\begin{align}
\dl t \| f^{\vpu{*}}_{rs} \dr+ \dl t \| f^{*}_{sr}\dr
&=i \sqrt{d+1} \bigl( Q_{rts} - Q_{srt} \bigr)
\nonumber
\\
& = i \sqrt{d+1} \left(
\frac{d+1}{d} \bigl(T_{rts}-T_{srt}\bigr) 
\right.
\nonumber
\\
& \hspace{1 in} \left. - 2\bigl(\dl t \| e_r\dr \dl e_r \| s\dr
-
\dl r\| e_s\dr \dl e_s \| t\dr
\vphantom{\frac{d}{d+1}}
\right)
\nonumber
\\
& = i \sqrt{\frac{2}{d}} \bigl( \dl t \| e_s \dr - \dl t \| e_r\dr\bigr)
\end{align}
where we used the fact that $T_{rts}=T_{srt}$ in the third step, and the fact that  $\dl t \| e_s\dr$ is real in the last.  This establishes Eq.~(\ref{eq:frsTermsfStsr}).  Eq.~(\ref{eq:fStrsTermsfsr}) is obtained by taking complex conjugates on both sides, and using the fact that the vectors $\| e_s \dr$ are real.  

Eqs.~(\ref{eq:QfrsIdA}) and~(\ref{eq:QfrsIdB}) are immediate consequences of the definitions, and the fact that $Q^{\vpu{T}}_{r}Q^{\mathrm{T}}_r=0$.  Turning to the proof of Eqs.~(\ref{eq:QfrsIdC}) and~(\ref{eq:QfrsIdD}), it follows from Eqs.~(\ref{eq:erEvecOfQr}) and~(\ref{eq:Qdef}) that
\be
Q_s \| e_s \dr = 0
\ee
Using this and the fact that $Q_s \| f^{*}_{sr}\dr = 0$ in Eq.~(\ref{eq:frsTermsfStsr}) we find
\be
Q^{\vpu{*}}_s \| f^{\vpu{*}}_{rs} \dr = - i \sqrt{\frac{2}{d}} Q^{\vpu{*}}_s \| e_r \dr
\ee
Since
\be
\| e_r \dr =\sqrt{\frac{d}{2(d+1)}} \Bigl( \| r\dr + \| \vv \dr\Bigr)
\label{eq:erTermsv0}
\ee
and taking account of the fact that $Q_s\| \vv\dr=0$ (see Eq.~(\ref{eq:vVecAnnihalation})) we deduce
\be
Q^{\vpu{*}}_s \| f^{\vpu{*}}_{rs} \dr = - i \sqrt{\frac{1}{d+1}} Q^{\vpu{*}}_s\| r\dr = 
-\frac{1}{d+1} \| f^{\vpu{*}}_{sr}\dr 
\ee
Taking complex conjugates on both sides of this equation we deduce the second identity in Eq.~(\ref{eq:QfrsIdC}).

In the same way, acting on both sides of Eq.~(\ref{eq:frsTermsfStsr}) with $Q^{\mathrm{T}}_s$ we find
\begin{align}
Q^{\mathrm{T}}_s \| f^{\vpu{*}}_{rs}\dr
& = 
- \| f^{*}_{sr}\dr -i \sqrt{\frac{2}{d}} Q^{\mathrm{T}}_s \| e_r \dr
\nonumber
\\
& = - \| f^{*}_{sr}\dr - i \sqrt{\frac{1}{d+1}} Q^{\mathrm{T}}_s \| r\dr
\nonumber
\\
& = - \frac{d}{d+1} \| f^{*}_{sr}\dr 
\end{align} 
Taking complex conjugates on both sides of this equation we deduce the second identity in Eq.~(\ref{eq:QfrsIdD}).

Turning to the last group of identities we have
\be
\dl f_{rs} \| f_{sr}\dr = \dl f_{rs}\| Q_r \| f_{sr} \dr = - \frac{1}{d+1} \dl f_{rs} \| f_{rs}\dr = - \frac{1}{d+1}
\ee
and
\be
\dl f^{\vpu{*}}_{rs} \| f^{*}_{sr} \dr = 
\dl f^{\vpu{*}}_{rs} \| Q^{\vpu{*}}_r \| f^{*}_{sr}\dr
=-\frac{d}{d+1} \dl f^{\vpu{*}}_{rs} \| f^{\vpu{*}}_{rs} \dr 
= - \frac{d}{d+1}
\ee
The other two identities are obtained by taking complex conjugates on both sides of the two just derived.
\end{proof}
This lemma provides a substantial part of what we need to prove the theorem.  The remaining part is provided by
\begin{lemma}
\label{lem:QQTids}
For all $r\neq s$
\begin{align}
Q^{\vpu{T}}_r Q{\vpu{T}}_s Q{\vpu{T}}_r &= \frac{1}{d+1} Q^{\vpu{T}}_{r} - \frac{d}{(d+1)^2} \| f^{\vpu{*}}_{rs}\dr \dl f^{\vpu{*}}_{rs} \|
\\
Q^{\vpu{T}}_r Q^{\mathrm{T}}_s Q^{\vpu{T}}_r & = \frac{d^2}{(d+1)^2} \| f^{\vpu{*}}_{rs} \dr \dl f^{\vpu{*}}_{rs} \|
\end{align}
\end{lemma}
\begin{proof}
It follows from Eq.~(\ref{eq:Qdef}) that
\begin{align}
Q_r Q_s Q_r & = \frac{d+1}{d}Q_r T_s Q_r - 2 Q_r \| e_s\dr \dl e_s \| Q_r 
\label{eq:QtripA}
\\
Q^{\vpu{T}}_r Q^{\mathrm{T}}_s Q^{\vpu{T}}_r & = \frac{d+1}{d} Q^{\vpu{T}}_r T^{\mathrm{T}}_s Q^{\vpu{T}}_r - 2 Q^{\vpu{T}}_r \| e^{\vpu{*}}_s\dr \dl e^{\vpu{*}}_s \| Q^{\vpu{T}}_r  
\label{eq:QtripB}
\end{align}
In view of Eqs.~(\ref{eq:erTermsv0}), (\ref{eq:vVecAnnihalation}) and the definition of $\| f_{rs}\dr$ we have
\be
Q_r \| e_s \dr  =\sqrt{\frac{d}{2(d+1)}} Q_r \|s \dr = - i \frac{\sqrt{d}}{\sqrt{2}(d+1)} \|f_{rs}\dr
\ee
Substituting this expression into Eqs.~(\ref{eq:QtripA}) and~(\ref{eq:QtripB}) we obtain 
\begin{align}
Q_r Q_s Q_r & = \frac{d+1}{d}Q_r T_s Q_r -\frac{d}{(d+1)^2} \| f_{rs}\dr \dl f_{rs} \|
\\
Q^{\vpu{T}}_r Q^{\mathrm{T}}_s Q^{\vpu{T}}_r & = \frac{d+1}{d} Q^{\vpu{T}}_r T^{\mathrm{T}}_s Q^{\vpu{T}}_r-\frac{d}{(d+1)^2} \| f_{rs}\dr \dl f_{rs} \|
\end{align}
The problem therefore reduces to showing
\begin{align}
Q_r T_s Q_r & = \frac{d}{(d+1)^2} Q_r 
\label{eq:QTQcalcA}
\\
Q^{\vpu{T}}_r T^{\mathrm{T}}_s Q^{\vpu{T}}_r & = \frac{d^2}{(d+1)^2} \| f^{\vpu{*}}_{rs}\dr \dl f^{\vpu{*}}_{rs}\| 
\label{eq:QTQcalcB}
\end{align}
Using Eq.~(\ref{eq:Qdef}) we find
\begin{align}
\dl a\|Q_r T_s Q_r\| b\dr & = \frac{(d+1)^2}{d^2} \dl a \| T_r T_s T_r\| b \dr
\nonumber
\\
&\hspace{0.25 in} - \frac{1}{2} \left(\frac{2(d+1)}{d}\right)^{\frac{3}{2}}\Bigl( K^2_{ra} \dl e_r \| T_s T_r \| b\dr+ K^2_{rb} \dl a \|  T_r T_s \| e_r \dr \Bigr)
\nonumber
\\
& \hspace{0.5 in}
  + \frac{2(d+1)}{d}K^2_{ra}K^2_{rb} \dl e_r \| T_s \| e_r  \dr 
\label{eq:QTQcalcC}
\\
\dl a\|Q^{\vpu{T}}_r T^{\mathrm{T}}_s Q^{\vpu{T}}_r\| b\dr & = \frac{(d+1)^2}{d^2} \dl a \| T^{\vpu{T}}_r T^{\mathrm{T}}_s T^{\vpu{T}}_r\| b \dr
\nonumber
\\
&\hspace{0.25 in} - \frac{1}{2} \left(\frac{2(d+1)}{d}\right)^{\frac{3}{2}}\Bigl( K^2_{ra} \dl e_r \| T^{\mathrm{T}}_s T^{\vpu{T}}_r \| b\dr+ K^2_{rb} \dl a \|  T^{\vpu{T}}_r T^{\mathrm{T}}_s \| e_r \dr \Bigr)
\nonumber
\\
& \hspace{0.5 in}
  + \frac{2(d+1)}{d}K^2_{ra}K^2_{rb} \dl e_r \| T^{\mathrm{T}}_s \| e_r  \dr 
\label{eq:QTQcalcD}
\end{align}
Using the definitions of $T_r$, $\|e_r\dr$ and Eq.~(\ref{eq:2designProp}) (the $2$-design property) we find, after some algebra,
\begin{align}
\dl a \| T_r T_s T_r \| b \dr 
&=
\frac{d^2}{(d+1)^2} \left(
K^2_{ra} T^{\vpu{2}}_{rsb} + K^2_{rb} T^{\vpu{2}}_{ras} + K^2_{rs} T^{\vpu{2}}_{rab} + K^2_{ra}K^2_{rb} 
\right)
\\
\dl e_r \| T_s T_r \| b\dr& = 2\left( \frac{d}{2(d+1)}\right)^{\frac{3}{2}}
\left(2K^2_{rs}K^2_{rb} + K^2_{rb} + T^{\vpu{2}}_{rsb} 
\right)
\\
\dl a\| T_r T_s \| e_r \dr & = 2\left( \frac{d}{2(d+1)}\right)^{\frac{3}{2}}
\left(2 K^2_{rs} K^2_{ra} + K^2_{ra} + T^{\vpu{2}}_{ras}
\right)
\\
\dl e_r \| T_s\| e_r \dr & = \frac{d}{2(d+1)} \left( 3 K^2_{rs} + 1
\right)
\end{align}
and
\begin{align}
\dl a \| T^{\vpu{T}}_r T^{\mathrm{T}}_s T^{\vpu{T}}_r \| b\dr
&=
\frac{d^2}{(d+1)^2} \Bigl(G_{ra}G_{as}G_{sb}G_{br}
\nonumber
\\
&\hspace{1.5 in} +K^2_{ra} T_{rsb} +K^2_{rb} T_{ras} + K^2_{ra}K^2_{rb} \Bigr)
\nonumber
\\
&=
\frac{d^2}{(d+1)^2} \Bigl( (d+1) T^{\vpu{2}}_{ras} T^{\vpu{2}}_{rsb} 
\nonumber
\\
& \hspace{1.5 in}  +
K^2_{ra} T^{\vpu{2}}_{rsb} + K^2_{rb}T^{\vpu{2}}_{ras}+K^2_{ra} K^2_{rb} 
\Bigr)
\label{eq:QTQcalcE}
\\
\dl e_r \|T^{\mathrm{T}}_s T^{\vpu{T}}_r  \| b\dr
&=
2\left(\frac{d}{2(d+1)}
\right)^{\frac{3}{2}}
\left( K^2_{rs} K^2_{rb} + K^2_{rb} + 2 T^{\vpu{2}}_{rsb}
\right)
\\
\dl a \| T^{\vpu{T}}_r T^{\mathrm{T}}_s \| e_r \dr 
& = 
2\left(\frac{d}{2(d+1)}
\right)^{\frac{3}{2}}
\left(K^2_{rs} K^2_{ra} + K^2_{ra} + 2 T^{\vpu{2}}_{ras}
\right)
\\
\dl e_r \| T^{\mathrm{T}}_s \| e_r \dr &= \frac{d}{2(d+1)} \left( 3 K^2_{rs}+1\right)
\end{align}
where in deriving  Eq.~(\ref{eq:QTQcalcE}) we used the fact that $G_{ra}G_{as}G_{sb}G_{br}$ $ =(d+1) T^{\vpu{2}}_{ras} T^{\vpu{2}}_{rsb}$ (in view of the fact that $r\neq s$). 
Substituting these expressions into Eqs.~(\ref{eq:QTQcalcC}) and~(\ref{eq:QTQcalcD}) we deduce Eqs.~(\ref{eq:QTQcalcA}) and~(\ref{eq:QTQcalcB}).
\end{proof}
Now define the rank $d-1$ projectors
\begin{align}
Q_{rs} & = Q_{r} - \| f_{rs}\dr \dl f_{rs} \|
\\
Q^{\mathrm{T}}_{rs} & = Q^{\mathrm{T}}_{r} - \| f^{*}_{rs} \dr \dl f^{*}_{rs} \|
\end{align}
and let $\mathcal{Q}^0_{rs}$, $\mathcal{Q}^{\vpu{0}}_{rs}$, $\bar{\mathcal{Q}}^0_{rs}$ and $\bar{\mathcal{Q}}^{\vpu{0}}_{rs}$ be, respectively, the subspaces onto which $\|f^{\vpu{*}}_{rs}\dr \dl f^{\vpu{*}}_{rs} \|$, $Q^{\vpu{*}}_{rs}$, $\| f^{*}_{rs}\dr \dl f^{*}_{rs} \|$ and $Q^{*}_{rs}$ project.  It is  immediate that we have the orthogonal decompositions
\begin{align}
\mathcal{Q}^{\vpu{0}}_r & = \mathcal{Q}^0_{rs} \oplus \mathcal{Q}^{\vpu{0}}_{rs}
\\
\bar{\mathcal{Q}}^{\vpu{0}}_r &= \bar{\mathcal{Q}}^0_{rs} \oplus \bar{\mathcal{Q}}^{\vpu{0}}_{rs}
\end{align} 
Using Lemma~\ref{lem:fVecIds}  we find
\be
Q_{sr} \| f_{rs} \dr = Q_{rs} \| f_{sr} \dr = 0
\ee
implying that $\mathcal{Q}^0_{rs} \perp \mathcal{Q}^{\vpu{0}}_{sr}$ and $\mathcal{Q}^{\vpu{0}}_{rs} \perp \mathcal{Q}^0_{sr}$, and  
\be
\bigl| \dl f_{rs} \| f_{sr} \dr \bigr| = \frac{1}{d+1}
\ee
implying that $\mathcal{Q}^0_{rs}$ and $\mathcal{Q}^0_{sr}$ are inclined at angle $\cos^{-1} \bigl( \frac{1}{d+1}\bigr)$.  Using  Lemma~\ref{lem:fVecIds} together with Lemma~\ref{lem:QQTids} we find
\begin{align}
Q_{rs} Q_{sr} Q_{rs} 
& = Q_{rs}Q_s Q_{rs} 
\nonumber
\\
& = Q_r Q_s Q_r - \| f_{rs} \dr \dl f_{rs} \| Q_s Q_r - Q_r Q_s \| f_{rs} \dr \dl f_{rs} \| 
\nonumber
\\
& \hspace{1.5 in} +
\dl f_{rs} \| Q_s \| f_{rs} \dr \| f_{rs} \dr \dl f_{rs} \|
\nonumber
\\
& = \frac{1}{d+1} Q_r -\frac{1}{d+1} \| f_{rs} \dr \dl f_{rs}\|
\nonumber
\\
& = \frac{1}{d+1} Q_{rs}
\end{align}
which in view of Lemma~\ref{lem:UniformIncline} implies that $\mathcal{Q}_{rs}$ and $\mathcal{Q}_{sr}$ are uniformly inclined at angle $\cos^{-1} \bigl( \frac{1}{\sqrt{d+1}}\bigr)$.  This proves part (a) of the theorem.  Parts (b) and (c) are proved similarly.

\subsection*{Proof of Theorem~\ref{thm:Rgeometry}}
Define 
\begin{align}
\| \gv^{\vpu{*}}_{rs} \dr & = \frac{1}{\sqrt{2}} \bigl(\| f^{*}_{rs} \dr + \| f^{\vpu{*}}_{rs} \dr \bigr)
\\
\| \bar{\gv}^{\vpu{*}}_{rs} \dr & = \frac{i}{\sqrt{2}} \bigl( \| f^{*}_{rs}\dr - \| f^{\vpu{*}}_{rs} \dr \bigr)
\end{align}
By construction the components of  $\| \gv_{rs} \dr$, $ \| \bar{\gv}_{rs} \dr$ in the standard basis are real, so we can regard them as $\in \mathbb{R}^{d^2}$.  They are orthonormal:
\be
\dl \gv_{rs} \| \gv_{rs} \dr = \dl \bar{\gv}_{rs} \| \bar{\gv}_{rs} \dr = 1 \qquad \text{and} \qquad \dl \gv_{rs} \| \bar{\gv}_{rs} \dr = 0 
\ee
It is also readily verified, using Lemma~\ref{lem:fVecIds}, that
\begin{align}
\bar{R}_r \| \gv_{rs} \dr &= \| \gv_{rs} \dr
\\
\bar{R}_r \| \bar{\gv}_{rs} \dr & = \| \bar{\gv}_{rs} \dr
\end{align}
So 
\be
R_{rs} = \bar{R}_r - \| \gv_{rs} \dr \dl \gv_{rs} \| - \| \bar{\gv}_{rs} \dr \dl \bar{\gv}_{rs}\|
\ee
is a rank $2d-4$ projector.  If we define $\mathcal{R}^0_{rs}$, $\mathcal{R}^1_{rs}$ and $\mathcal{R}^{\vpu{0}}_{rs}$ to be, respectively, the subspaces onto which $\| \gv_{rs} \dr \dl \gv_{rs} \|$, $\| \bar{\gv}_{rs} \dr \dl \bar{\gv}_{rs} \|$ and $R_{rs}$ project we have the orthogonal decomposition
\be
\mathcal{R}^{\vpu{0}}_{r} = \mathcal{R}^{0}_{rs} \oplus \mathcal{R}^{1}_{rs} \oplus \mathcal{R}^{\vpu{0}}_{rs}
\ee
It follows from Eqs.~(\ref{eq:frsTermsfStsr}) and~(\ref{eq:fStrsTermsfsr}) that
\be
\| \gv_{rs} \dr = - \| \gv_{sr} \dr
\ee
 implying that $\mathcal{R}^0_{rs} = \mathcal{R}^0_{sr}$ for all $r\neq s$.  It is also easily verified, using Lemma~\ref{lem:fVecIds}, that
\be
\bigl| \dl \bar{\gv}_{rs} \| \bar{\gv}_{sr} \dr \bigr| = \frac{d-1}{d+1}
\ee
from which it follows that $\mathcal{R}^1_{rs}$ and $\mathcal{R}^1_{sr}$ are inclined at angle $\cos^{-1} \bigl( \frac{d-1}{d+1} \bigr)$.  We next observe that
\be
R^{\vpu{T}}_{rs} = Q^{\vpu{T}}_{rs} + Q^{\mathrm{T}}_{rs} 
\ee
Using Lemma~\ref{lem:fVecIds} once again we deduce 
\be
R_{rs} \|\bar{\gv}_{sr}\dr  = R_{sr} \| \bar{\gv}_{rs} \dr = 0 
\ee
from which it follows that $\mathcal{R}^1_{rs} \perp \mathcal{R}^{\vpu{1}}_{sr}$ and $\mathcal{R}^{\vpu{1}}_{rs} \perp \mathcal{R}^{1}_{sr}$.  Finally, we know from Theorem~\ref{thm:Qgeometry} that $Q^{\mathrm{T}}_{rs}Q^{\vpu{T}}_{sr}= Q^{\vpu{T}}_{rs} Q^{\mathrm{T}}_{sr} = 0$.  Consequently
\begin{align}
R^{\vpu{T}}_{rs} R^{\vpu{T}}_{sr} R^{\vpu{T}}_{rs} &= Q^{\vpu{T}}_{rs} Q^{\vpu{T}}_{sr} Q^{\vpu{T}}_{rs} + Q^{\mathrm{T}}_{rs} Q^{\mathrm{T}}_{sr} Q^{\mathrm{T}}_{rs}
\nonumber
\\
& =\frac{d}{d+1} Q^{\vpu{T}}_{rs} +\frac{d}{d+1} Q^{\mathrm{T}}_{rs} 
\nonumber
\\
& = \frac{1}{d+1} R^{\vpu{T}}_{rs}
\end{align}
In view of   Lemma~\ref{lem:UniformIncline} it follows that $\mathcal{R}_{rs}$ and $\mathcal{R}_{sr}$ are uniformly inclined at angle $\cos^{-1} \bigl( \frac{1}{\sqrt{d+1}}\bigr)$.

\subsection*{Further Identities}
We conclude this section with another set of identities involving the vectors $\| f^{\vpu{*}}_{rs}\dr$, $\|f^{*}_{rs}\dr$, $\|\gv^{\vpu{*}}_{rs}\dr$ and $\|\bar{\gv}^{\vpu{*}}_{rs}\dr$.

Define
\be
\| \ebv_r \dr = \sqrt{\frac{2d}{d-1} } \| \ev_r\dr - \sqrt{\frac{d+1}{d-1}} \| \vv \dr 
\ee
where $\| \vv \dr$ is the vector defined by Eq.~(\ref{eq:vZeroDef}). It is readily verified that
\be
\dl \ebv_r \| \ebv_r \dr  = 0
\qquad \text{and} \qquad
\dl \ebv_r \| \vv \dr  = 0
\ee
So $\|\ebv_r\dr$, $\| \vv \dr$ is an orthonormal basis for the $2$-dimensional subspace spanned by $\|\ev_r\dr$, $\|\vv\dr$.  Note that
\be
Q^{\vpu{T}}_r \| \ebv_r \dr = Q^{\mathrm{T}}_r \| \ebv_r \dr = \bar{R}^{\vpu{T}}_r\| \ebv_r \dr = 0
\ee
We then have
\begin{theorem}
For all $r$
\begin{align}
\frac{1}{d+1} \sum_{\substack{s=1\\ (s\neq r)}}^{d^2} \| f_{rs} \dr \dl f_{rs} \| &= Q_r
\label{eq:fvecFinalIdsA}
\\
\frac{1}{d+1} \sum_{\substack{s=1\\ (s\neq r)}}^{d^2} \| f^{*}_{rs} \dr \dl f^{*}_{rs} \| &= Q^{\mathrm{T}}_r
\label{eq:fvecFinalIdsB}
\\
\frac{2}{d+1} \sum_{\substack{s=1\\ (s\neq r)}}^{d^2} \| \gv_{rs} \dr \dl \gv_{rs} \| &= \bar{R}_r
\label{eq:fvecFinalIdsC}
\\
\frac{2}{d+1} \sum_{\substack{s=1\\ (s\neq r)}}^{d^2} \| \bar{\gv}_{rs} \dr \dl \bar{\gv}_{rs} \| &= \bar{R}_r
\label{eq:fvecFinalIdsD}
\end{align}
and
\begin{align}
\frac{1}{d-1} \sum_{\substack{s=1\\ (s\neq r)}}^{d^2} \| f^{\vpu{*}}_{sr} \dr \dl f^{\vpu{*}}_{sr} \| &= Q^{\mathrm{T}}_r + \| \ebv_r\dr \dl \ebv_r \| +\frac{1}{d^2-1}\Bigl(I - \| \vv \dr \dl \vv \|
\Bigr)
\label{eq:fvecFinalIdsE}
\\
\frac{1}{d-1} \sum_{\substack{s=1\\ (s\neq r)}}^{d^2} \| f^{*}_{sr} \dr \dl f^{*}_{sr} \| &= Q^{\vpu{*}}_r + \| \ebv^{\vpu{*}}_r\dr \dl \ebv^{\vpu{*}}_r \| +\frac{1}{d^2-1}\Bigl(I - \| \vv \dr \dl \vv \|
\Bigr)
\label{eq:fvecFinalIdsF}
\\
\frac{2}{d+1}\sum_{\substack{s=1\\ (s\neq r)}}^{d^2} \| \gv_{sr}\dr \dl \gv_{sr} \| & = \bar{R}_{r}
\label{eq:fvecFinalIdsG}
\\
\frac{2}{d-3}\sum_{\substack{s=1\\ (s\neq r)}}^{d^2} \| \bar{\gv}_{sr}\dr \dl \bar{\gv}_{sr} \| & =
 \bar{R}_r + \frac{4(d-1)}{d-3} \| \ebv_r \dr \dl \ebv_r \| +\frac{4}{(d+1)(d-3)}\Bigl(I - \| \vv \dr \dl \vv \| \Bigr)
\label{eq:fvecFinalIdsH}
\end{align}
\end{theorem}
\begin{proof}
It follows from the definition of $\| f^{\vpu{*}}_{rs}\dr$ that
\begin{align}
\frac{1}{d+1}\sum_{\substack{s=1\\ (s\neq r)}}^{d^2} \| f_{rs} \dr \dl f_{rs} \| 
&=
\sum_{\substack{s=1\\ (s\neq r)}}^{d^2} Q_r \| s\dr \dl s \| Q_r
\nonumber
\\
& = Q_r \left(\sum_{s=1}^{d^2} \| s\dr \dl s \| \right) Q_r
\nonumber
\\
& = Q_r
\end{align}
where in the second step we used the fact that $Q_r \| r\dr=0$ (as can be seen by setting $r=s$ in  Eq.~(\ref{eq:Qexplicit})).  Eq.~(\ref{eq:fvecFinalIdsB}) is obtained by taking the complex conjugate on both sides.

We also have
\begin{align}
\frac{1}{d+1}\sum_{\substack{s=1\\ (s\neq r)}}^{d^2} \| f^{\vpu{*}}_{rs} \dr \dl f^{*}_{rs} \| 
&=
-\sum_{\substack{s=1\\ (s\neq r)}}^{d^2} Q^{\vpu{T}}_r \| s\dr \dl s \| Q^{\mathrm{T}}_r
\nonumber
\\
& = - Q^{\vpu{T}}_r \left(\sum_{s=1}^{d^2} \| s\dr \dl s \| \right)Q^{\mathrm{T}}_r
\nonumber
\\
& = - Q^{\vpu{T}}_r Q^{\mathrm{T}}_r
\nonumber
\\
& = 0
\end{align}
Taking the complex conjugate on both sides we find
\be
\frac{1}{d+1}\sum_{\substack{s=1\\ (s\neq r)}}^{d^2} \| f^{*}_{rs}  \dr \dl f^{\vpu{*}}_{rs}\| = 0
\ee
Consequently
\begin{align}
\frac{2}{d+1} \sum_{\substack{s=1\\ (s\neq r)}}^{d^2} \| \gv_{rs} \dr \dl \gv_{rs} \| &= 
\frac{1}{d+1} \sum_{\substack{s=1\\ (s\neq r)}}^{d^2} \Bigl(\| f_{rs} \dr \dl f_{rs} \|
+\| f^{*}_{rs} \dr \dl f^{*}_{rs} \|
\nonumber
\\
& \hspace{ 1 in }+ \| f^{\vpu{*}}_{rs} \dr \dl f^{*}_{rs} \| +\| f^{*}_{rs}  \dr \dl f^{\vpu{*}}_{rs}\|
\Bigr)
\nonumber
\\
& = \bar{R}_r
\end{align}
Eq.~(\ref{eq:fvecFinalIdsD}) is proved similarly.

To prove the second group of identities we have to work a little harder.  Using Eqs.~(\ref{eq:erVecDef}) and~(\ref{eq:Qdef}) we find
\begin{align}
\frac{1}{d-1} \sum_{\substack{s=1\\ (s\neq r)}}^{d^2} \dl a \| f_{sr} \dr \dl f_{sr} \| b \dr
& = \frac{d+1}{d-1}\sum_{s=1}^{d^2} \dl a \| Q_s\| r\dr \dl r \| Q_s \| b\dr
\nonumber
\\
& = \frac{(d+1)^3}{d^2(d-1)}\sum_{s=1}^{d^2} \Bigl(T_{sar}T_{srb}- K^2_{sa}K^2_{sr}T^{\vpu{2}}_{srb}
\nonumber
\\
& \hspace{1 in} -K^2_{sr}K^2_{sb}T^{\vpu{2}}_{sar}+K^2_{sa}K^4_{sr}K^2_{sb}
\Bigr)
\label{eq:fvecFinalIdInterA}
\end{align}
(where we used the fact that $Q_s \| s\dr = 0$ in the first step).  After some algebra we find
\begin{align}
\sum_{s=1}^{d^2} T_{sar}T_{srb} & = \frac{d}{d+1} \Biggr(
\left(\sqrt{\frac{d-1}{d+1}  }\dl a \| \ebv_r \dr+\frac{1}{d} 
\right)
\left(\sqrt{\frac{d-1}{d+1}  }\dl \ebv_r\| b \dr+\frac{1}{d} 
\right)
\nonumber
\\
& \hspace{2.5 in}
 + T^{\vpu{2}}_{rba}\Biggl)
\\
\sum_{s=1}^{d^2} K^2_{sa} K^2_{sr}T^{\vpu{2}}_{srb}&=
\frac{d}{d+1}
\Biggl(
\left(\sqrt{\frac{d-1}{d+1}}\dl a\| \ebv_r \dr +\frac{2d+1}{d(d+1)}
\right)
\left(\sqrt{\frac{d-1}{d+1}}\dl \ebv_r\| b\dr +\frac{1}{d} \right)
\nonumber
\\
& \hspace{2.5 in} +\frac{1}{d+1} T_{rba}
\Biggr)
\\
\sum_{s=1}^{d^2} K^2_{sr} K^2_{sb}T^{\vpu{2}}_{sar}&=
\frac{d}{d+1}
\Biggl(
\left(\sqrt{\frac{d-1}{d+1}}\dl a\| \ebv_r \dr +\frac{1}{d}
\right)
\left(\sqrt{\frac{d-1}{d+1}}\dl \ebv_r\| b\dr +\frac{2d+1}{d(d+1)} \right)
\nonumber
\\
& \hspace{2.5 in} +\frac{1}{d+1} T_{rba}
\Biggr)
\\
\sum_{s=1}^{d^2} K^2_{sa}K^4_{sr}K^2_{sb}
&=
\frac{d}{(d+1)} \Biggl(
\frac{d+2}{d+1} \left(
\sqrt{\frac{d-1}{d+1}}\dl a \| \ebv_r \dr + \frac{1}{d} 
\right)\left(
\sqrt{\frac{d-1}{d+1}}\dl \ebv_r \| b \dr + \frac{1}{d} 
\right)
\nonumber
\\
&  \hspace{1.65 in} + \frac{d}{(d+1)^3} \delta_{ab} + \frac{d+2}{(d+1)^3} \Biggr)
\end{align}
where we used Eq.~(\ref{eq:2designProp}) to derive the first  expression.  Substituting these expressions into Eq.~(\ref{eq:fvecFinalIdInterA}) and using 
\be
\dl a \| Q^{\mathrm{T}}_r \| b\dr = \frac{d+1}{d}\left( T_{rba} - 
\left(\sqrt{\frac{d-1}{d+1}}\dl a \| \ebv_r\dr + \frac{1}{d}  
\right) 
\left(\sqrt{\frac{d-1}{d+1}}\dl \ebv_r \| b\dr + \frac{1}{d}  
\right)
\right)
\ee
we deduce Eq.~(\ref{eq:fvecFinalIdsE}).  Taking  complex conjugates on both sides we obtain Eq.~(\ref{eq:fvecFinalIdsF}).

Eq.~(\ref{eq:fvecFinalIdsG}) is an immediate consequence of Eq.~(\ref{eq:fvecFinalIdsC}) and the fact that $\| \gv_{sr} \dr = -\| \gv_{rs} \dr$ for all $r,s$.

To prove Eq.~(\ref{eq:fvecFinalIdsH}) observe that it follows from Eqs.~(\ref{eq:fvecFinalIdsE})--(\ref{eq:fvecFinalIdsG}) that
\begin{align}
\sum_{\substack{s=1\\(s\neq r)}}^{d^2} \Bigl(\| f^{\vpu{*}}_{sr} \dr \dl f^{*}_{sr} \|
+\| f^{*}_{sr} \dr \dl f^{\vpu{*}}_{sr} \| \Bigr)
& = \sum_{\substack{s=1\\(s\neq r)}}^{d^2}\Bigr( 2 \| \gv_{sr} \dr \dl \gv_{sr} \|
-\| f^{\vpu{*}}_{sr} \dr \dl f^{\vpu{*}}_{sr} \| - \| f^{*}_{sr} \dr \dl f^{*}_{sr} \|\Bigr)
\nonumber
\\
&= 2 \left(\bar{R}_r -(d-1) \| \ebv_r \dr \dl \ebv_r \| \vpu{\frac{1}{d+1}} \right.
\nonumber
\\
& \hspace{1 in} \left. - \frac{1}{d+1}\Bigl(I - \| \vv \dr \dl \vv \| \Bigr) \right)
\end{align}
Hence
\begin{align}
\frac{2}{d-3} \sum_{\substack{s=1\\(s\neq r)}}^{d^2} \| \bar{\gv}_{sr}\dr \dl \bar{\gv}_{sr} \|
&=\frac{1}{d-3} \sum_{\substack{s=1\\(s\neq r)}}^{d^2} \Bigl(\| f^{\vpu{*}}_{sr} \dr \dl f^{\vpu{*}}_{sr} \| + \| f^{*}_{sr} \dr \dl f^{*}_{sr} \| 
\nonumber
\\
& \hspace{1.25 in} - \| f^{\vpu{*}}_{sr}\dr \dl f^{*}_{sr}\|
- \| f^{*}_{sr}\dr \dl f^{*}_{sr} \| \Bigr)
\nonumber
\\
& =\bar{R}_r + \frac{4(d-1)}{d-3}\|\ebv_r \dr \dl \ebv_r \| + \frac{4}{(d+1)(d-3)}\Bigl( I - \| \vv \dr \dl \vv \| \Bigr)
\end{align}

\end{proof}

\section{The \texorpdfstring{$\ggrm$-$\ggrm^{\mathrm{T}}$}{P-PT} Property}
\label{sec:GGstarProp}

In the preceding sections the $Q$-$Q^{\mathrm{T}}$ property has played a prominent role.  In this section we show that in the particular case of a  Weyl-Heisenberg covariant SIC-POVM, and with the appropriate choice of gauge, the Gram projector (defined in Eq.~(\ref{eq:SICGramProjectorDef})) has an analogous property, which we call the $\ggrm$-$\ggrm^{T}$ property.  Specifically one has
\be
\ggrm \ggrm^{\mathrm{T}} = \ggrm^{\mathrm{T}}\ggrm  = \| \hv \dr \dl \hv \|
\ee
where $\| \hv\dr$ is a normalized vector whose components in the standard basis are all real.  In odd dimensions the components of $\| h \dr $ in the standard basis can be simply expressed in terms of the Wigner function of the fiducial vector.   It could be said that the projectors $\ggrm$ and $\ggrm^{\mathrm{T}}$ are almost orthogonal (by contrast with the projectors $Q^{\vpu{T}}_r$ and $Q^{\mathrm{T}}_r$ which are completely orthogonal).  More precisely $\ggrm$ has the spectral decomposition
\be
\ggrm = \bar{\ggrm} + \| \hv \dr \dl \hv \|
\ee
where $\bar{\ggrm}$ is a rank $(d-1)$ projector with the property
\be
 \bar{\ggrm}^{\vpu{T}} \bar{\ggrm}^{\mathrm{T}}= 0
\ee
This means that the matrix
\be
\grmJ= \ggrm-\ggrm^{\mathrm{T}}
\ee
is a pure imaginary Hermitian matrix with the property that $\grmJ^2$ is a real rank $2d-2$ projector (\emph{c.f.} the discussion in  Section~\ref{sec:QQTpropGen}). 

Although we are mainly interested in the $\ggrm$-$\ggrm^{\mathrm{T}}$ property as it applies to SIC-POVMs, it should be noted that it actually holds for any Weyl-Heisenberg covariant POVM (with the appropriate choice of gauge).  So we will prove the above propositions for this more general case.  

Let us begin by fixing notation.  Let $|0\rangle, \dots, |d-1\rangle$ be an orthonormal basis for $d$-dimensional Hilbert space and let $X$ and $Z$ be the operators
whose action on the $|r\rangle$ is
\begin{align}
X | a\rangle & = |a+1\rangle
\\
Z |a\rangle & = \omega^{a} |a\rangle
\end{align}
where $\omega=e^{\frac{2  \pi i}{d}}$ and the addition of indices in the first equation is \emph{mod}  $d$.  We then define the Weyl-Heisenberg displacement operators by
(adopting the convention used in, for example, ref.~\cite{selfA})
\be
D_{\mathbf{p}} = \tau^{p_1p_2} X^{p_1} Z^{p_2}
\label{eq:whDisOpDef}
\ee
where $\mathbf{p}$ is the  vector $(p_1,p_2)$ ($p_1$, $p_2$ being integers) and $\tau= e^{\frac{(d+1)\pi i}{d}}$.  Generally speaking the decision to insert the phase $\tau^{p_1p_2}$ is a matter of convention, and many authors define it differently, or else omit altogether.  However, for the purposes of this section it is essential, as a different choice of phase at this stage would lead to a different gauge in the class of POVMs to be defined below, and the Gram projector would then typically not have the $\ggrm$-$\ggrm^{\mathrm{T}}$ property.    

Note that $\tau^2 = \tau^{d^2}= \omega$ in every dimension.  If the dimension is odd we can write $\tau = \omega^{\frac{d+1}{2}}$.  So $\tau$ is a $d^{\mathrm{th}}$ root of unity.  However, if the dimension is even $\tau^d = -1$. This  has the consequence that 
\be
D_{\mathbf{p}+d\mathbf{u}} = (-1)^{u_1 p_2 + u_2 p_1} D_{\mathbf{p}}
\ee
So in even dimension $\mathbf{p}=\mathbf{q} \text{ (mod $d$)}$ does not necessarily imply $D_{\mathbf{p}}=D_{\mathbf{q}}$ (although the  operators are, of course, equal if $\mathbf{p}=\mathbf{q} \text{ (mod $2d$)}$)

In every dimension  (even or odd) we have
\begin{align}
D^{\dagger}_{\mathbf{p}} &= D^{\vpu{\dagger}}_{-\mathbf{p}}
\label{eq:whPropertyA}
\\
\intertext{for all $\mathbf{p}$}
\left(D_{\mathbf{p}}\right)^{n} & = D_{n\mathbf{p}}
\label{eq:whPropertyB}
\\
\intertext{for all $\mathbf{p}$, $n$ and}
D_{\mathbf{p}} D_{\mathbf{q}} & = \tau^{\langle \mathbf{p},\mathbf{q}\rangle} D_{\mathbf{p}+\mathbf{q}}
\label{eq:whPropertyC}
\end{align}
for all $\mathbf{p},\mathbf{q}$.  In the last expression $\langle \mathbf{p},\mathbf{q} \rangle$ is the symplectic form
\be
\langle \mathbf{p},\mathbf{q} \rangle = p_2 q_1 - p_1 q_2
\ee

Now let $|\psi\rangle$ be any normalized vector (not necessarily a SIC-fiducial vector), and define
\be
|\psi_{\mathbf{p}}\rangle = D_{\mathbf{p}}|\psi\rangle
\ee
Let 
\be
L = \sum_{\mathbf{p}\in\mathbb{Z}_d^2} |\psi_\mathbf{p}\rangle \langle \psi_{\mathbf{p}}|
\ee
It is easily seen that
\be
\bigl[D_{\mathbf{p}} , L\bigr] = 0
\ee
for all $\mathbf{p}$.  

We now appeal to the fact that there is no non-trivial subspace of $\mathcal{H}_d$ which the displacement operators leave invariant.   To see this assume the contrary.  Then there would exist non-zero vectors $|\phi\rangle$, $|\chi\rangle$ such that
\be
\langle \phi | D_{\mathbf{p}} |\chi\rangle  = 0
\ee
for all $\mathbf{p}$.  Writing the left-hand side out in full this gives
\be
\sum_{a=0}^{d-1} \omega^{p_2 a} \langle \phi | a+ p_1 \rangle \langle a |\chi\rangle = 0 
\ee
for all $p_1, p_2$.  Taking the discrete Fourier transform with respect to $p_2$, we have
\be
\langle \phi | a + p_1 \rangle \langle a | \chi \rangle = 0
\ee
for all $a$, $p_1$, implying that either $|\phi\rangle = 0$ or $|\chi \rangle = 0$---contrary to assumption.
We can therefore use Schur's lemma~\cite{Lie4} to deduce  that 
\be
L = k I
\ee
for some constant $k$.  Taking the trace on both sides of this equation we infer that $k = d$.  We conclude that $\frac{1}{d} |\psi_{\mathbf{p}}\rangle \langle \psi_{\mathbf{p}}|$ is  a POVM.  We refer to POVMs of this general class as Weyl-Heisenberg covariant POVMs.  We refer to the vector $|\psi\rangle$ which generates the POVM as the fiducial vector (with no implication that it is necessarily a SIC-fiducial).

Now consider the Gram projector 
\begin{align}
\ggrm & = \sum_{\mathbf{p},\mathbf{q} \in \mathbb{Z}^2_d} P_{\mathbf{p},\mathbf{q}} \| \mathbf{p} \dr \dl \mathbf{q} \| 
\\
\intertext{where}
\ggrm_{\mathbf{p},\mathbf{q}}& = \frac{1}{d} \langle\psi_{\mathbf{p}}|\psi_{\mathbf{q}}\rangle 
\end{align}
and where we label the matrix elements of $\ggrm$ and the standard basis kets with  the vectors $\mathbf{p}$, $\mathbf{q}$ rather than with the single integer indices $r,s$ as in the rest of this paper.  We know from Theorem~\ref{thm:GramSICcondition} that $\ggrm$ is a rank $d$ projector.

 In view of Eqs.~(\ref{eq:whPropertyA}) and~(\ref{eq:whPropertyC}) we have
\begin{align}
 \dl \mathbf{p} \| P \|\mathbf{q} \dr & = \ggrm_{\mathbf{p},\mathbf{q}}
\nonumber
\\
& =  \frac{1}{d}\tau^{-\langle \mathbf{p},\mathbf{q}\rangle}\langle \psi | D_{\mathbf{q}-\mathbf{p}}|\psi \rangle
\nonumber
\\
& = \frac{1}{d}  \sum_{a=0}^{d-1} \tau^{p_1p_2 + q_1 q_2}\omega^{a q_2 - (q_1+a)p_2 }\langle \psi | a+ q_1-p_1 \rangle \langle a | \psi\rangle
\end{align}
Hence
\begin{align}
\dl \mathbf{p} \| P P^{\mathrm{T}} \| \mathbf{q} \dr
& = \sum_{\mathbf{u}\in\mathbb{Z}_d} \dl \mathbf{p} \| P \| \mathbf{u} \dr \dl \mathbf{q} \| P \| \mathbf{u} \dr
\nonumber
\\
& =\frac{1}{d^2}  \sum_{a,b,u_1,u_2=0}^{d-1} \tau^{p_1p_2+q_1q_2}\omega^{u_2(u_1+a+b) - (u_1+a)p_2-(u_1+b)q_2}
\nonumber
\\
& \hspace{1 in} \times
\langle \psi | a+u_1 -p_1\rangle \langle \psi | b+ u_1 - q_1 \rangle \langle a | \psi \rangle \langle b | \psi \rangle
\nonumber
\\
& = \frac{1}{d} \sum_{a,b=0}^{d-1} \tau^{p_1p_2+q_1q_2} \omega^{p_2 b+q_2 a}\langle \psi| -b - p_1\rangle \langle b |\psi\rangle \langle \psi | -a - q_1\rangle \langle a|\psi\rangle 
\nonumber
\\
& = \dl \mathbf{p} \| \hv \dr \dl \hv \| \mathbf{q} \dr
\end{align}
where $\| \hv \dr$ is the vector with components
\be
\dl \mathbf{p} \| \hv \dr = \frac{1}{\sqrt{d}} \sum_{a=0}^{d-1} \tau^{p_1p_2} \omega^{p_2 a} \langle \psi | - a - p_1 \rangle \langle a| \psi \rangle 
\ee
It is easily verified that $\| h \dr$ is normalized, and that $\dl \mathbf{p} \|\hv\dr$ is real.  

Finally, suppose that the dimension is odd.  Then the Wigner function of the state $|\psi\rangle$ is~\cite{Wootters,Vourdas}
\be
W(\mathbf{p}) = \frac{1}{d} \langle \psi | D^{\vpu{\dagger}}_{\mathbf{p}} \ptyD D^{\dagger}_{\mathbf{p}} |\psi\rangle = 
\frac{1}{d} \langle \psi | D_{2 \mathbf{p}} \pty | \psi \rangle 
\ee
where $\pty$ is the parity operator, whose action on the standard basis is $ \pty |a\rangle = |-a \rangle $.
It is straightforward to show
\be
\dl \mathbf{p} \| \hv \dr = \sqrt{d} W(-2^{-1} \mathbf{p}) 
\ee
where $2^{-1}=(d+1)/2$ is the multiplicative inverse of $2$ considered as an element of $\mathbb{Z}_d$:  \emph{i.e.} the unique integer  $0\le m <d$ such that $2 m = 1 \text{ (mod $d$)}$.

\section{Conclusion}

A curious fact about SIC-POVMs is that, although they are characterized by their being highly symmetric structures, they do not wear this property on their sleeve (so to speak).  If one casually inspects the components of a SIC-fiducial, without knowing in advance that that is what they are, there does not seem to be anything special about them at all.  Indeed, so far from there being any obvious pattern to the components, they seem, to a casual inspection, like a completely random collection of numbers.  Moreover, this is just as true of an exact fiducial as it is of a numerical one (see, for instance, the tabulations in Scott and Grassl~\cite{ScottGrassl}).  It is only when one looks at them through the right pair of spectacles, and takes the trouble to calculate the overlaps $\Tr(\Pi_r \Pi_s)$, that the symmetry becomes apparent.  The situation is a little reminiscent of a hologram, which only takes on the aspect of a meaningful image when it is viewed in the right way.  If one wanted to summarize the content of this paper  in a nutshell it could be said that we have exhibited some other pairs of spectacles---other ways of looking at a SIC---which cause its inner secrets (or at any rate some of its inner secrets) to become manifest. 

Rather than focusing on the SIC-vectors $|\psi_r\rangle$, as is usually done, we have focused on the angle tensors $\theta_{rs}$ and $\theta_{rst}$, and on the $T$, $J$ and $R$ matrices defined in terms of them.  This is an important change of emphasis because, rather than being tied to any particular SIC, these quantities characterize entire families of unitarily equivalent SICs.   Like the components of a SIC-fiducial, the angle tensors appear, to a casual inspection, like a random collection of numbers.  However, if one examines the spectra of the $T$, $J$ and $R$ matrices one realizes that, underlying the appearance of randomness, there is a high degree of order.  If one then goes on to examine the geometrical relationships between the subspaces onto which the $Q$, $Q^{\mathrm{T}}$ and $\bar{R}$ matrices project, as we did in Section~\ref{sec:geometry}, one finds yet more instances of structure and order.  To our minds what is particularly interesting about all of this is that none of it is obviously suggested by the defining property of a SIC, that $\Tr(\Pi_r\Pi_s) = 1/(d+1)$ for $r\neq s$.  

In the course of this paper we have several times expressed the hope that the Lie algebraic perspective on a SIC will lead to a solution to the existence problem.  Of course, that is only a hope, and it may not be fulfilled.  However, we feel on rather safer ground when we suggest that the solution is likely to come, if not from this  investigation, then from one which is like it to the extent that it focuses on a feature of a SIC which is not immediately apparent to the eye.  

Specializing to the case of a Weyl-Heisenberg covariant SIC, a fiducial vector $|\psi\rangle$ is a solution to the equations
\be
\bigl| \langle \psi | D_{\mathbf{p}} | \psi \rangle \bigr|^2 = \frac{d\delta_{\mathbf{p},\boldsymbol{0}}+1}{d+1}
\ee
Allowing for the arbitrariness of the overall phase of $|\psi\rangle$, and taking $|\psi\rangle$ to be normalized, this gives us  $d^2-1$ conditions on only $2d-2$ independent real parameters.  The equations are thus over-determined, and very highly over-determined when $d$ is large.  Nevertheless, they have turned out to be soluble in every case which has been investigated to date. It seems likely that progress will depend on finding  the structural feature which is responsible for this remarkable fact.   The motivation for this paper is the belief that it may be structural features of the Lie algebra $\gl(d,\mathbb{C})$ which are responsible.  That suggestion may  or may not be correct.  But if it turns out to be incorrect, the amount of effort which has been expended on this problem over a period of more than ten years, so far without fruit, suggests to us that the solution will depend on finding  some other structural feature of a SIC, which is not obvious, and which has hitherto escaped attention. 
}

\section{Acknowledgements}

The authors thank I. Bengtsson for discussions.  Two authors, DMA and CAF, were supported in part by the U.~S. Office of Naval Research (Grant No.\ N00014-09-1-0247). Research at Perimeter Institute is supported by the Government of Canada through Industry Canada and by the Province of Ontario through the Ministry of Research \& Innovation.


\end{document}